\documentclass[onecolumn]{IEEEtran}
\usepackage{amsmath,amsthm,amssymb}
\usepackage{lineno}
\usepackage{geometry}
\usepackage{color}
\usepackage{listings}
\usepackage{algorithm}
\usepackage{algorithmic}
\usepackage{setspace}
\usepackage{multirow}
\usepackage{diagbox} % 加载宏包
\usepackage{tikz}
\usetikzlibrary{trees,arrows}
\usepackage{makecell}
\usepackage{float}
\usepackage{ulem}
\usepackage{textgreek}

\newtheorem{theorem}{Theorem}[section]
\newtheorem{lemma}{Lemma}[section]
\newtheorem{corollary}{Corollary}[section]
\newtheorem{proposition}{Proposition}[section]
\newtheorem{define}{Definition}[section]
\newtheorem{remark}{Remark}[section]

\newtheorem{example}{Example}[section]
\def\R{{\mathbb{R}}}

\lstdefinelanguage{Maple}{
   keywords={if, while, do, else, end, for, from, to,then},
   keywordstyle=\color{blue}\bfseries,
   ndkeywords={class, export, boolean, throw, implements, import, this},
   ndkeywordstyle=\color{darkgray}\bfseries,
   identifierstyle=\color{black},
   sensitive=false,
   comment=[l]{//},
   morecomment=[s]{/*}{*/},
   commentstyle=\color{purple}\ttfamily,
   stringstyle=\color{red}\ttfamily,
   morestring=[b]',
   morestring=[b]"
}

\lstdefinelanguage{SOStools}{
   keywords={syms,sosprogram,monomials,sosineq,sossetobj,sossolve,sosgetsol,sospolyvar},
   keywordstyle=\color{blue}\bfseries,
   ndkeywords={syms,sosprogram,monomials,sosineq,sossetobj,sossolve,sosgetsol},
   ndkeywordstyle=\color{blue}\bfseries,
   identifierstyle=\color{black},
   sensitive=false,
   comment=[l]{//},
   morecomment=[s]{/*}{*/},
   commentstyle=\color{purple}\ttfamily,
   stringstyle=\color{red}\ttfamily,
   morestring=[b]',
   morestring=[b]"
}

\ifCLASSINFOpdf
\else
\fi
\begin{document}

\title{Proving Information Inequalities and Identities with Symbolic Computation}


\author{Laigang~Guo,~
        Raymond~W.~Yeung,~%\IEEEmembership{Fellow,~IEEE,}
        and~Xiao-Shan~Gao%,~\IEEEmembership{Member,~IEEE}% <-this % stops a space
\thanks{L. Guo is with the Laboratory of Mathematics and Complex Systems (Ministry of Education), School of Mathematical Sciences, Beijing Normal University, Beijing, China. e-mail: (lgguo@bnu.edu.cn).}% <-this % stops a space
\thanks{R. W. Yeung is with the Institute of Network Coding and the Department of Information Engineering, The Chinese
University of Hong Kong, N.T., Hong Kong. e-mail: (whyeung@ie.cuhk.edu.hk).}% <-this % stops a space
\thanks{X.-S. Gao is with the Key Laboratory of Mathematics Mechanization, Institute of Systems Science, AMSS, Chinese Academy of Sciences, and University of Chinese Academy of Sciences, Beijing, China. e-mail: (xgao@mmrc.iss.ac.cn).}
%\thanks{Manuscript received April 19, 2005; revised September 17, 2014.}
}

\maketitle

%\doublespacing

\begin{abstract}
\noindent

Proving linear inequalities and identities of Shannon's information measures, possibly with linear constraints on the information measures,
is an important problem in information theory. For this purpose, ITIP and other variant algorithms have been developed and implemented, which are all based on solving a linear program (LP). 
In particular, an identity $f = 0$ is verified by solving two LPs, one for $f \ge 0$ and one for $f \le 0$.
In this paper, we develop a set of algorithms that can be implemented by symbolic computation. Based on these algorithms, procedures for verifying linear information
inequalities and identities are devised. 
Compared with LP-based algorithms, our procedures can produce analytical proofs that 
are both human-verifiable and free of numerical errors. Our procedures are also more efficient computationally. For constrained inequalities, by taking advantage of the algebraic structure of the problem, the size of the LP that needs to be solved can be significantly
reduced. For identities, instead of solving two LPs, the identity can be verified directly with very little computation.

\end{abstract}


\IEEEpeerreviewmaketitle

\section{Introduction}

%\IEEEPARstart{S}{hannon's}
In information theory, we may need to prove various information inequalities and identities that involve Shannon's information measures. For example, such information inequalities and identities play a crucial role in establishing the converse of most coding theorems. However, proving an information inequality or identity  
involving more than a few random variables can be highly non-trivial.

To tackle this problem, a framework for linear information inequalities
was introduced in \cite{Yeung1997}. Based on this framework, the problem of 
verifying Shannon-type inequalities can be formulated as a linear program (LP), and a software package based on MATLAB called Information Theoretic Inequality Prover (ITIP) was developed \cite{Yeung.Yan1996}.
Subsequently, different variations of ITIP have been developed. Instead of MATLAB,
Xitip \cite{Pulikkoonattu2006} uses a C-based linear programming solver, and it has  been further developed into its web-based version, oXitip \cite{Rathenakar2020}. 
minitip \cite{Csirmaz2016} is a C-based version of ITIP that adopts a simplified syntax
and has a user-friendly syntax checker.
psitip \cite{Li2020} is a Python library that can verify unconstrained/constrained/existential entropy inequalities. It is a computer algebra system where random variables, expressions, and regions are objects that can be manipulated. 
AITIP \cite{Ho2020} is a cloud-based platform that not only provides analytical proofs for Shannon-type inequalities but also give hints on constructing a smallest counterexample in case the inequality to be verified is not a Shannon-type inequality.

Using the above LP-based approach, to prove an information identity $f = 0$, two LPs need to be solved, one for the inequality $f \ge 0$ and the other for the inequality $f \le 0$. Roughly speaking, the amount of computation for proving an information identity is twice the amount for proving an information inequality. If the underlying random variables exhibit certain Markov or functional dependence structures, there exist more efficient approaches to proving information identities \cite{Yeung2019}\cite{Chan2019}.

The LP-based approach is in general not computationally efficient because it does not take advantage of the special structure of the underlying LP. In this paper, we take a different approach. Instead of transforming the problem into a general LP to be solved numerically, we develop algorithms that can implemented by symbolic computation,
and based on these algorithms, procedures for proving information inequalities and identities are devised. Our specific contributions are:
\begin{enumerate}
\item 
Analytical proofs for information inequalities and identities that are free of numerical errors can be produced.
\item 
Compared with the LP-based approach, the computational efficiency of our procedure is in general much higher.
\item 
Information identities can be proved directly with very little computation instead of having to solve 2 LPs.
\end{enumerate}

The rest of the paper is organized as follows. 
In Section~II, we present the preliminaries for information inequalities. 
In Section~III, we develop algorithms for simplifying a set of linear inequalities
subject to linear inequality and equality constraints.
In Section~IV, we introduce a set of variables (inspired by the theory of $I$-Measure \cite{Yeung1991})
that facilitates the implementation of our algorithms. 
In Section~V, the procedures for proving information inequalities and identities are presented. Two examples are given in Section~VI to illustrate our procedures. Section~VII concludes the paper.

%%%------------------------------------------------------------------------------------------------------------------------------------------

%\section{Preliminaries}
%\label{sec-Pre}

%\subsection{Basic inequalities and elemental inequalities}
\section{Information inequality preliminaries}
In this section, we present some basic results related to information inequalities and their verification. For a comprehensive discussion on the topic, we refer the reader to 
\cite[Chs.~13-15]{Yeung2008}.

It is well known that all Shannon's information measures, namely entropy, conditional entropy, mutual information, and conditional mutual information are always nonnegative. The nonnegativity of all Shannon's information measures forms a set of
inequalities called the {\it basic inequalities}. The set of basic inequalities, however, is not minimal in the sense that some basic inequalities are implied by the others. For example,
$$H(X|Y)\geq0\ {\rm and}\ I(X;Y)\geq0,$$
which are both basic equalities involving random variables $X$ and $Y$, imply
$$H(X)=H(X|Y)+I(X;Y)\geq0,$$
again a basic equality involving $X$ and $Y$.
In order to eliminate such redundancies, the minimal subset of the basic inequalities
was found in \cite{Yeung1997}. 

Throughout this paper, all random variables are discrete. Unless otherwise specified, all 
information expressions involve some or all of the random variables $X_1,X_2, \ldots,X_n$. The value of $n$ will be specified when necessary.  Denote the set $\{1,2,\ldots,n\}$ by $\mathcal{N}_n$ and the sequence $[1,2,\ldots,n]$ by $[n]$.

\begin{theorem}{\rm \cite{Yeung1997}}
\label{element-ine}
Any Shannon's information measure can be expressed as a conic combination of the following two elemental forms of
Shannon's information measures:

i) $H(X_i|X_{\mathcal{N}_n-\{i\}})$

ii) $I(X_i;X_j|X_K)$, where $i\neq j$ and $K\subseteq \mathcal{N}_n-\{i,j\}$.

\end{theorem}

The nonnegativity of the two elemental forms of Shannon’s information measures forms a proper subset of the set of basic inequalities. The inequalities in this smaller set are called the {\it elemental inequalities}.
In \cite{Yeung1997}, the minimality of the elemental inequalities is also proved.
%that is, there is no redundant inequalities in the elemental inequalities.
%
The total number of elemental inequalities is equal to 
\begin{equation*}
m=n+\sum\limits_{r=0}^{n-2}\left(
  \begin{array}{c}
    n\\
    r\\
  \end{array}
\right)
\left(
  \begin{array}{c}
    n-r\\
    2\\
  \end{array}
\right)
=n+
\left(
  \begin{array}{c}
    n\\
    2\\
  \end{array}
\right)  2^{n-2}  .
\label{m}
\end{equation*}

In this paper, inequalities (identities) involving only Shannon's information measures
are referred to as information inequalities (identities). 
The elemental inequalities are called {\it unconstrained} information inequalities
because they hold for all joint distributions of the random variables.
In information theory, we very often deal with information inequalities (identities) that hold  under certain constraints on the joint distribution of the random variables. These are called {\it constrained} information inequalities (identities), and the associated constraints are usually  expressible as linear constraints on the Shannon's information measures. We will 
confine our discussion on constrained inequalities of this type.

\begin{example}
\def\ra{\rightarrow}
The celebrated data processing theorem asserts that for any four random variables 
$X$, $Y$, $Z$ and $T$, if $X \rightarrow Y \rightarrow Z \rightarrow T$ 
forms a Markov chain, then
$I(X;T) \ge I(Y;Z)$. Here, $I(X;T) \ge I(Y;Z)$ is a constrained information inequality under the constraint $X \rightarrow Y \rightarrow Z \rightarrow T$, which is equivalent to 
\[
\left\{
\begin{array}{rcl}
I(X;Z|Y) & = & 0 \\
I(X,Y;T|Z) & = & 0 ,
\end{array}
\right.
\]
or 
\[
I(X;Z|Y) + I(X,Y;T|Z) = 0
\]
owing to the nonnegativity of conditional mutual information. Either way, the Markov
chain can be expressed a set of linear constraint(s) on the Shannon's information measures.	
\end{example}

Information ineaualities (unconstrained or constrained) that are implied by the 
basic inequalities are called {\it Shannon-type} inequalities. Most of the information 
inequalities that are known belong this type. However, {\it non-Shannon-type} 
inequalities do exist, e.g., \cite{Zhang1998}. See \cite[Ch.~15]{Yeung2008} for a discussion.
	
Shannon's information measures, with conditional mutual informations 
being the general form, can be expressed as a linear combination of joint entropies by means of following identity:
$$ I(X_G;X_{G'}|X_{G''})=H(X_G,X_{G''})+H(X_{G',G''})-H(X_G,X_{G'},X_{G''})-H(X_{G'}). $$
where $G,G',G''\subseteq \mathcal{N}_n$. 
For the random variables $X_1, X_2, \ldots, X_n$, there are a total of $2^n-1$ joint entropies.
By regarding the joint entropies as variables, the basic (elemental) inequalities
become linear inequality constraints in $\R^{2^n-1}$. By the same token, 
the linear equality constrains on 
Shannon's information measures imposed by the problem under discussion become
linear equality constraints in $\R^{2^n-1}$. This way, the problem of verifying 
a (linear) Shannon-type inequality can be formulated as a linear program (LP),
which is described next.

Let $\bf{h}$ be the column $m$-vector of the joint entropies 
of $X_1, X_2, \ldots, X_n$. The set of elemental inequalities can be written as 
$G {\bf h} \ge 0$, where $G$ is an $m \times (2^n-1)$ matrix and $G {\bf h} \ge 0$ means
all the components of $G {\bf h}$ are nonnegative. Likewise, the constraints on 
the joint entropies can be written as $Q {\bf h} = 0$. When there is no constraint on the joint entropies, $Q$ is assumed to have zero row. The following theorem enables
a Shannon-type inequality to be verified by solving an LP.

\begin{theorem} {\rm \cite{Yeung1997}}
${\bf b}^\top {\bf h} \ge 0$ is a Shannon-type inequality under the constraint $Q {\bf h} = 0$ 
if and only if the minimum of the problem
\begin{center}
Minimize ${\bf b}^\top {\bf h}$, subject to $G {\bf h} \ge 0$ and $Q {\bf h} = 0$
\end{center}
is zero. 
\label{LP-S}
\end{theorem}

%-----------------------------------------------------------------------------------------------------------------------------------------
%%%%%-------------------------------------------------------------------------------------------------------------------------

%%%%%%------------------------------------------------------------------------------------------------------------------------------------

%\subsection{Preliminary knowledge of polynomials}
\section{Linear inequalities and related algorithms}
\label{LI-algo.}

In this section, we will develop some algorithms for simplifying a linear inequality set constrained by a linear equality set. 
These algorithms will be used as building blocks for the procedures 
to be developed in Section \ref{procedures} for proving information inequalities and identities.

We will start by discussing some notions pertaining to linear inequality sets and linear
equality sets. Then we will establish some related properties that are instrumental for 
developing the aforementioned algorithms. 

Let $\mathbf{x}=[x_1,x_2,\ldots,x_n]$, and let $\mathbb{R}_h[\mathbf{x}]$ be the set of all homogeneous linear polynomials in $\mathbf{x}$ with real coefficients.
In this paper, unless otherwise specified, we assume that all inequality sets have the form $S_f=\{f_i\geq0,i\in\mathcal{N}_m\}$, with $f_i\neq0$ and $f_i\in\mathbb{R}_h[\mathbf{x}]$, and all the equality sets have the form $E_{\tilde{f}}=\{\tilde{f}_i=0,i\in\mathcal{N}_{\widetilde{m}}\}$ with $\tilde{f}_{i}\neq0$ and $\tilde{f}_i\in\mathbb{R}_h[\mathbf{x}]$.

For a given set of polynomials $P_f=\{f_i,i\in\mathcal{N}_m\}$ and the corresponding set of inequalities $S_f=\{f_i\geq0,i\in\mathcal{N}_m\}$, and a given set of polynomials $P_{\tilde{f}}=\{\tilde{f}_i,i\in\mathcal{N}_{\widetilde{m}}\}$ and the corresponding set of equalities $E_{\tilde{f}}=\{\tilde{f}_i=0,i\in\mathcal{N}_{\widetilde{m}}\}$, where $f_i$ and $\tilde{f}_i$ are polynomials in $\mathbf{x}$, we write
%{\color{red}
$S_f=\mathcal{R}(P_f)$, $P_f=\mathcal{R}^{-1}(S_f)$,
$E_{\tilde{f}}=\widetilde{\mathcal{R}}(P_{\tilde{f}})$ and $P_{\tilde{f}}=\widetilde{\mathcal{R}}^{-1}(E_{\tilde{f}})$.
% 
%{\color{blue}
%Let Span$_{\mathbb{R}}(P_f)$ be an $\mathbb{R}-$linear vector space generated by $P_f$.}

\begin{define}\label{subset}
Let $S_f=\{f_i\geq0,i\in\mathcal{N}_m\}$ and $S_{f'}=\{f'_i\geq0,i\in\mathcal{N}_{m'}\}$ be two inequality sets, and $E_{\tilde{f}}$ and $E_{\tilde{f}'}$ be two equality sets.  We write $S_{f'}\subseteq S_{f}$ if $\mathcal{R}^{-1}(S_{f'})\subseteq \mathcal{R}^{-1}(S_{f})$, and $E_{\tilde{f}'}\subseteq E_{\tilde{f}}$ if $\widetilde{\mathcal{R}}^{-1}(E_{\tilde{f}'})\subseteq \widetilde{\mathcal{R}}^{-1}(E_{\tilde{f}})$. 
Furthermore, we write $(f_i\geq0) \in S_f$ to mean that the inequality $f_i\geq0$ is included in $S_f$.
\end{define}

\begin{define}
Let $N_{>0}=\{1,2,\ldots\}$. For $a_i\in N_{>0}, i\in\mathcal{N}_n$, a sequence $[a_1,a_2,...,a_{n}]$ is said to be in {\it descending order} if $a_1\geq a_2\geq\cdots\geq a_{n}$.
\end{define}

%%%%------------------------------------------------------------------------------------------------------------------------------------

\begin{define}
Let $\mathbb{R}_{>0}$ and $\mathbb{R}_{\geq0}$ be the sets of positive and nonnegative real numbers, respectively.
A linear polynomial $F$ in $\mathbf{x}$ is called a {\it positive (nonnegative) linear combination} of polynomials $f_j$ in $\mathbf{x}$, $j=1,\ldots,k$, if $F=\sum_{j=1}^{k}r_jf_j$ with $r_j\in$ $\mathbb{R}_{>0}$ $(r_j\in$ $\mathbb{R}_{\geq0})$.
%\footnote{Here, the set of positive real numbers is denoted by $\mathbb{R}_{>0}$, and the set of nonnegative real numbers is denoted by $\mathbb{R}_{\geq0}$}.
%
A nonnegative linear combination is also called a {\it conic} combination.
\end{define}

%\begin{define}
%Let Span$_{\mathbb{R}}(\widetilde{\mathcal{R}}^{-1}(\widetilde{E}))$ be an $\mathbb{R}-$linear vector space generated by $\widetilde{\mathcal{R}}^{-1}(\widetilde{E})$. 
%\end{define}

\begin{define}
The inequalities $f_1\geq0,f_2\geq0,\ldots,f_k\geq0$ imply the inequality $f\geq0$ if the following holds:
\begin{center}
$\mathbf{x}$ satisfying $f_1\geq0,f_2\geq0,\ldots,f_k\geq0$ implies $\mathbf{x}$ satisfies $f\geq0$.
\end{center}

\end{define}

\begin{define}
Given a set of inequalities $S_f=\{f_i\geq0,i\in\mathcal{N}_m\}$, for some $i\in\mathcal{N}_m$, $f_i\geq0$ is called a redundant inequality if $f_i\geq0$ is implied by the inequalities $f_j\geq0$, where $j\in \mathcal{N}_m$ and $j\neq i$.
%or equivalently, $f_i$ is a conic combination of $f_j\in \mathcal{R}^{-1}[S_f],\ j\neq i$.

\end{define}

%\begin{define}\label{Ineq-set}
%Given two sets of inequalities $S_f=\{f_i\geq0,i\in\mathcal{N}_m\}$ and $S_g=\{g_i\geq0,i\in\mathcal{N}_m\}$, we say that

%1) the inequality $f_0\geq0$ is included in $S_{f}$ if there exists $c_0\in\mathbb{R}_{\geq0}$ such that $f_0=c_0\cdot f_{i_0}$ for some $i_0\in\mathcal{N}_m$.

%2) $S_f=S_g$ if for any $f_{i_1}\in \mathcal{R}^{-1}[S_f]$, we can find $g_{i_2}\in \mathcal{R}^{-1}[S_{g}]$ such that $g_{i_2}=c_1\cdot f_{i_1}$, with $c_1\in\mathbb{R}_{\geq0}$, and for any $g_{i_3}\in \mathcal{R}^{-1}[S_g]$, we can find $f_{i_4}\in \mathcal{R}^{-1}[S_{f}]$ such that $f_{i_4}=c_2\cdot g_{i_3}$, with $c_2\in\mathbb{R}_{\geq0}$.

%\end{define}

%{\color{blue}
\begin{define}\label{Ineq-set}
Two inequalities $f\geq0$ and $g\geq0$ are trivially equivalent if $f=c\,g$ for some $c\in \mathbb{R}_{>0}$.
Given two sets of inequalities $S_f=\{f_i\geq0,i\in\mathcal{N}_{m_1}\}$ and $S_g=\{g_i\geq0,i\in\mathcal{N}_{m_2}\}$, we say that
$S_f$ and $S_g$ are trivially equivalent if 
\begin{enumerate}
%\item[•] $S_f$ and $S_g$ have exactly the same number of inequalities,
%\item[•] for every $i\in\mathcal{N}_{m_1}$, $f_i\geq0$ is trivially equivalent to $g_j\geq0$ for some $j\in\mathcal{N}_{m_2}$, 
%\item[•] for every $i\in\mathcal{N}_{m_2}$, $g_i\geq0$ is trivially equivalent to $f_j\geq0$ for some $j\in\mathcal{N}_{m_1}$.
\item $S_f$ and $S_g$ have exactly the same number of inequalities;
\item for every $i\in\mathcal{N}_{m_1}$, $f_i\geq0$ is trivially equivalent to $g_j\geq0$ for some $j\in\mathcal{N}_{m_2}$;
\item for every $i\in\mathcal{N}_{m_2}$, $g_i\geq0$ is trivially equivalent to $f_j\geq0$ for some $j\in\mathcal{N}_{m_1}$.
\end{enumerate}
Furthermore, if $S_f$ and $S_g$ are trivially equivalent, then we regard $S_f$ and $S_g$ as the same set of inequalities.

%Based on the above notations, we use $(f_i\geq0) \in S_f$ to mean the inequality $f_i\geq0$ is included in $S_f$.
\end{define}
%}

\begin{lemma}[Farkas' Lemma\cite{Farkas1902,Achiya1997}]
	\label{Farkas}
	Let $\mathbf{A}\in \mathbb{R}^{m\times n}$ and $\mathbf{b}\in\mathbb{R}^{m}$. Then exactly one the following two assertions is true:
	
	1. There exists an ${\rm \mathbf{x}}\in\mathbb{R}^n$ such that $\mathbf{A}{\rm \mathbf{x}}=\mathbf{b}$ and ${\rm \mathbf{x}}\ge0$.
	
	2. There exists a ${\rm \mathbf{y}}\in\mathbb{R}^m$ such that $\mathbf{A}^{T}{\rm \mathbf{y}}\geq0$ and $\mathbf{b}^{T}\mathbf{y}<0$.
	
\end{lemma}

\begin{lemma}\label{GG}
	Given $h_1,\ldots,h_k,{h}\in\mathbb{R}_h[\mathbf{y}]$,  $h_1\ge0,...,h_k\ge0$ imply ${h}\ge0$ if and only if ${h}$ is a conic combination of $h_1,\ldots,h_k$.
	
\end{lemma}

\begin{proof}
	It is straightforward that $h_1\ge0,\ldots,h_k\ge0$ imply ${h}\ge0$ if ${h}$ is a conic combination of $h_1,\ldots,h_k$.
	We need only to prove the converse.

	Assume that $h_1\ge0,\ldots,h_k\ge0$.
	Define a vector $\mathbf{h}=(h_1,\ldots,h_k)^T$, and the variable vector $\mathbf{y}=(y_1,\ldots,y_m)^T$.
	Since $h_1,\ldots,h_k,{h}\in\mathbb{R}_h[\mathbf{y}]$, we can let
	$\mathbf{h} = \mathbf{A}^{T}\mathbf{y}$ and ${h} = \mathbf{b}^{T}\mathbf{y}$, where $\mathbf{A}\in \mathbb{R}^{m\times n}$ and $\mathbf{b}\in\mathbb{R}^{m}$.
	Since $h_1\ge0,...,h_k\ge0$ imply ${h}\ge0$, there exists no ${\rm \mathbf{y}}\in\mathbb{R}^m$ such that $\mathbf{A}^{T}{\rm \mathbf{y}}\geq0$ and $\mathbf{b}^{T}\mathbf{y}<0$, which means Assertion 2 in Lemma \ref{Farkas} is false.
	Then by the lemma, Assertion 1 must be true, that is, there exists an ${\rm \mathbf{x}}\in\mathbb{R}^n$ such that $\mathbf{A}{\rm \mathbf{x}}=\mathbf{b}$ and ${\rm \mathbf{x}}\ge0$. Then we have
	$$
	\mathbf{A}{\rm \mathbf{x}}=\mathbf{b} \Rightarrow (\mathbf{A}{\rm \mathbf{x}})^T=\mathbf{b}^T
	\Rightarrow {\rm \mathbf{x}}^T\mathbf{A}^T=\mathbf{b}^T
	\Rightarrow {\rm \mathbf{x}}^T\mathbf{A}^T\mathbf{y}=\mathbf{b}^T\mathbf{y}
	\Rightarrow {\rm \mathbf{x}}^T\mathbf{h}={h},
	$$
	which implies that ${h}$ is a conic combination of $h_1,\ldots,h_k$.
	The lemma is proved.
	
\end{proof}

Note that this lemma generalizes Theorem~2 in \cite{Yeung1997}. 

\begin{define}\label{def-purein}
Let $S_f=\{f_i(\mathbf{x})\geq0,i\in\mathcal{N}_{m}\}$ be an inequality set.
	%If $f_{k}(\mathbf{x})$ satisfies the condition in  Lemma \ref{equa1}, 
	If $f_{k}(\mathbf{x})=0$ for all solution $\mathbf{x}$ of $S_f$,
	then $f_{k}(\mathbf{x})=0$ is called an {\it implied equality} of $S_f$. The inequality set $S_f$ is called a {\it pure inequality set} if $S_f$ has no implied equalities.
\end{define}

%{\color{blue}
	\begin{lemma}\label{equa1}
		Let $S_f=\{f_i(\mathbf{x})\geq0,i\in\mathcal{N}_{m}\}$ be an inequality set.
		Then $f_k$ is an implied inequality of $S_f$ if and only if 
		\begin{equation}
			f_{k}(\mathbf{x})\equiv\sum\limits_{i=1,i\neq k}^{m}p_{i}f_{i}(\mathbf{x}),
		\label{qqoinva}
		\end{equation}
	where $p_i \le 0$ for all $i \in {\cal N}_m \backslash \{k\}$.
	\end{lemma}
	
	\begin{proof}
		Assume (\ref{qqoinva}) holds and let ${\bf x}$ be any solution of $S_f$.
		Then $f_{k}(\mathbf{x})=\sum\limits_{i=1,i\neq k}^{m}p_{i}f_{i}(\mathbf{x})\leq0$ since $p_i\leq0$ and $f_{i}(\mathbf{x})\geq0$, for $i\in\mathcal{N}_m \backslash \{k\}$. On the other hand, from $f_{k}(\mathbf{x})\geq0$, we obtain $f_{k}(\mathbf{x})=0$.
		Therefore, $f_k({\bf x}) = 0$ for all solution ${\bf x}$ of $S_f$, 
		i.e., $f_k$ is an implied equality of $S_f$.
		
		Now, assume that $f_k$ is an implied inequality of $S_f$, i.e., $f_{k}(\mathbf{x})=0$ for all solution $\mathbf{x}$ of $S_f$.
		This implies that if $\mathbf{x}$ is a solution of $S_f$, then 
		$f_k(\mathbf{x}) \le 0$. In other words, the inequality $f_k(\mathbf{x}) \le 0$
		is implied by the $S_f$. 
		By Lemma \ref{GG}, there exist $q_i \ge 0$, $i \in {\cal N}_m$ such that  
		$$
		-f_{k}(\mathbf{x})\equiv\sum\limits_{i=1}^{m}q_{i}f_{i}(\mathbf{x}).
		$$
		Then,
		$$
		(-1-q_k)f_{k}(\mathbf{x})\equiv\sum\limits_{i=1,i\neq k}^{m}q_{i}f_{i}(\mathbf{x}), $$
		or 
		$$f_{k}(\mathbf{x})\equiv\sum\limits_{i=1,i\neq k}^{m}\left(-\frac{q_{i}}{1+q_k}\right)f_{i}(\mathbf{x}) .
		$$
		Upon letting $p_i=-\frac{q_{i}}{1+q_k}$, where $p_i \le 0$ since $q_i \ge 0$, we obtain (\ref{qqoinva}). This completes the proof.		
	\end{proof}

Let $E_{\bar{f}}$ be the set of all implied equalities of $S_f$. Evidently, $\widetilde{\mathcal{R}}^{-1}(E_{\bar{f}})\subseteq\mathcal{R}^{-1}(S_f)$. 
%As a simple example, if $S_f$ contains both  $f_s\ge0$ and $-f_s\ge0$, then $f_s=0$ and %$-f_s=0$ are implied equalities.
Next, we give an example to show that if an equality set is imposed,
a pure inequality set can become a non-pure inequality set.

\begin{example}
Let $S_f=\{f_1\geq0,f_2\geq0\}$, where $f_1=x_1+x_2$, $f_2=x_1-x_2$. Evidently, $S_f$ is a pure inequality set. However, if we impose the constraint $x_1=0$, then $S_f$ becomes $\{x_2\geq0,-x_2\geq0\}$, which is a non-pure inequality set. 
\end{example}

\begin{proposition}\label{subsetofpureispure}
A subset of a pure inequality set is a pure inequality set.
\end{proposition}

\begin{proof}
The proposition follows immediately from Lemma~\ref{equa1} and Definition~\ref{def-purein}.
\end{proof}
%\begin{proof}
%Let ${S}_g=\{g_i\geq0,i\in\mathcal{N}_m\}$ be a pure inequality set, and ${S}_{g'}=\{g'_i\geq0,i\in\mathcal{N}_{m'}\}$ be a 
%%minimal characterization set
%subset
%of $S_{g}$. %Then we know that $S_{g'}\subseteq S_{g}$. 
%Assume there exists an implied equality $u=0$ of $S_{g'}$. By Lemma \ref{equa1} and Definition \ref{def-purein}, we have
%$u\in \mathcal{R}^{-1}(S_{g'})$, and $u=\sum\limits_{j=1}^{k}p_{j}g'_{i_j}$ where $p_{j}\leq0$, $i_j\in\mathcal{N}_{m'}$, and $g'_{i_j}\in \mathcal{R}^{-1}(S_{g'})$. By Definition \ref{subset}, we have $u\in\mathcal{R}^{-1}(S_g)$ and $g'_{i_j}\in\mathcal{R}^{-1}(S_g)$, which mean
%$u=0$ is an implied equality of $S_{g}$.
%This contradicts that $S_g$ is a set of pure inequality set. 
%% Proving the proposition.
%%Thus we say if there is no implied equality in $S_{g}$, then there is no implied equality in $S_{g'}$. In other words, a minimal characterization set of a pure inequality set is also a pure inequality set.
%\end{proof}

\begin{define}
Let $S_f=\{f_i\geq0,i\in\mathcal{N}_m\}$ and $S_{f'}=\{f'_i\geq0,i\in\mathcal{N}_{m'}\}$ be two inequality sets.
%
%{\color{red}
If the solution sets of  $S_{f'}$ and $S_f$ are the same,  then we say that $S_{f}$ and $S_{f'}$ are equivalent.
\end{define}

\begin{proposition}\label{Prop2}
If $S_f$ and $S_{f'}$ are equivalent, then every inequality in $S_f$ is implied by $S_{f'}$, and every inequality in $S_{f'}$ is implied by $S_{f}$.
\end{proposition}

In the rest of the section, we will develop 
a few algorithms for simplifying a linear inequality set constrained  
by a linear equality set. 
\subsection{Dimension Reduction of a set of inequalities by an equality set}

Let $S_f=\{f_i\ge 0, i\in\mathcal{N}_m\}$ be an inequality set and $E_{\tilde{f}}=\{\tilde{f}_i=0,i\in\mathcal{N}_{\widetilde{m}}\}$ be an equality set. 
Recall that
$P_f=\mathcal{R}^{-1}(S_f)=\{f_i,i\in\mathcal{N}_m\}$ and 
$P_{\tilde{f}}=\widetilde{\mathcal{R}}^{-1}(E_{\tilde{f}})=\{\tilde{f}_i,i\in\mathcal{N}_{\widetilde{m}}\}$.
%
%Obviously, we know that $P_{\tilde{f}}\subseteq P_{f}$ and $\widetilde{m}\leq m$.
%
The following proposition is well known (see for example \cite[Chapter 1]{Lay2016}).
%{\color{red}
\begin{proposition}\label{rsb-def}
Under the variable order $x_1\prec x_2\prec \cdots\prec x_n$,
 the linear equation system  $E_{\tilde{f}}$
 %$E= \widetilde{\mathcal{R}}(P_{\tilde{f}})$ 
 can be reduced by Gauss-Jordan elimination to the unique form 
\begin{equation}\label{Sta-Basis}
\widetilde{E}=\{x_{k_i}-U_i=0,i\in \mathcal{N}_{\widetilde{n}}\},
\end{equation}
where $k_1<k_2<\cdots<k_{\widetilde{n}}$, $x_{k_i}$ is the leading term of $x_{k_i}-U_i$, $\widetilde{n}$ is rank of the linear system $E_{\tilde{f}}$ and $U_i$ is a linear function in 
$\{x_j, \hbox{ for }k_i<j<k_{i+1}, i\in \mathcal{N}_{\widetilde{n}}\}$, with $k_{i+1}=n+1$ by convention. 
Furthermore, $\sum_{i\in\mathcal{N}_{\tilde{n}}}|U_i|=n-\tilde{n}$.
\end{proposition}

\begin{algorithm}[H]
\caption{Dimension Reduction}
\label{division-algorithm}
\begin{algorithmic}[1]

\REQUIRE $S_f$, $E_{\tilde{f}}$.
\ENSURE The remainder set $R_f$.
\STATE Compute $\widetilde{E}$ with $E_{\tilde{f}}$ by Proposition \ref{rsb-def}.
\STATE Substitute $x_{k_i}$ by $U_i$ in $P_{f}$ to obtain a set $R$.
\STATE Let $R_f=R\backslash\{0\}$. 

\RETURN ${\cal R}(R_f)$.
\end{algorithmic}
\end{algorithm}

We call the equality set $\widetilde{E}$ the {\it Jordan normal form} of $E_{\tilde{f}}$. Likewise, we call the polynomial set $\widetilde{R}^{-1}(\widetilde{E})$ the Jordan normal form of 
$\widetilde{R}^{-1}(E_{\tilde{f}})$.
We say reducing $S_f$ by $E_{\tilde{f}}$ to mean using Algorithm~\ref{division-algorithm} to find ${\cal R}(R_f)$.
We also say reducing $P_f$ by $E_{\tilde{f}}$
to mean using Algorithm~\ref{division-algorithm} to find $R_f$, called {\it the remainder set} (or remainder if $R_f$ is a singleton).

\begin{example}
Given a variable order $x_1\prec x_2\prec x_3$, let $S_{f}=\{f_1\ge0,f_2\ge0\}$ and $E_{\tilde{f}}=\{\tilde{f}_1=0,\tilde{f}_2=0,\tilde{f}_3=0\}$, where $f_1=x_1+x_2-x_3$, $f_2=x_2+x_3$, $\tilde{f}_1=x_1+x_2+x_3$, $\tilde{f}_2=x_1+x_2$, and $\tilde{f}_3=x_3$. We write $P_f=\mathcal{R}^{-1}(S_f)=\{f_1,f_2\}$ and $P_{\tilde{f}}=\widetilde{\mathcal{R}}^{-1}(E_{\tilde{f}})=\{\tilde{f}_1,\tilde{f}_2,\tilde{f}_3\}$.

Firstly, we obtain that the rank of $E_{\tilde{f}}$ is $\tilde{n}=2$. Then the
Jordan normal form of $E_{\tilde{f}}$ is given by $\widetilde{E}=\{x_{k_1}-U_1=0,x_{k_2}-U_2=0\}$, where $k_1=1,\ k_2=3$, $U_1=-x_2$, $U_2=0$. 

Using the equality constraints in $\tilde{E}$, we substitute $x_1=-x_2$ and $x_3=0$ into $P_f = \{ f_1, f_2 \}$ to obtain $R=\{0,x_2\}$.
Hence
%the remainder set on division of $P_f$ by $\widetilde{\mathcal{R}}^{-1}(\widetilde{E})$ is 
$R_f=R\backslash\{0\}=\{x_2\}$.
In other words, the inequality set $S_f$ is reduced to $\mathcal{R}(R_f)=\{x_2\geq0\}$ by the equality set $E_{\tilde{f}}$. Note that in $\mathcal{R}(R_f)$, only $n-\tilde{n}=1$ variable, namely $x_2$, appears.

\end{example}

\begin{remark}
\label{4398hqf}
After the execution of Algorithm~\ref{division-algorithm},
the inequality set $S_f$ constrained by the equality set $E_{\tilde{f}}$
is reduced to the inequality set $\mathcal{R}(R_f)$ constrained by the equality set $\tilde{E}$. Therefore, the solution set of `$S_f$ constrained by $E_{\tilde{f}}$' 
in $\mathbb{R}^n$ is the same as the solution set of `$\mathcal{R}(R_f)$ constrained by $\tilde{E}$' in $\mathbb{R}^n$.
\end{remark}

%\begin{remark}
%\label{qatrpwj}
%Note that $n$ variables are involved in $S_f$, 
%but only $n-\tilde{n}$ variables are involved in ${\cal R}(R_f)$.
%Thus ${\cal R}(R_f)$ may as well be regarded as an inequality set in 
%$\mathbb{R}^{n-\tilde{n}}$.
%Therefore, if we want to optimize the value of some $F \in \mathbb{R}_h({\bf x})$
%over the solution set of $S_f$, we only need to optimize the value of $F'$, 
%where $F'$ is obtained by reducing $F$ by $E_{\tilde{f}}$, over the solution set of 
%${\cal R}^{-1}(R_f)$ in $\mathbb{R}^{n-\tilde{n}}$. Moreover, for this purpose, the equality constraints $E_{\tilde{f}}$ (or equivalent $\tilde{E}$)
%can be ignored. As a result, the required computation can be significantly reduced. 
%This will be exploited in Section~\ref{procedures} when we develop the procedures for proving information inequalities and identities.
%\end{remark}

\subsection{The implied equalities contained in a system of inequalities}
\label{sec-I}

In this subsection, we will show how to find all the implied equalities contained in a system of linear inequalities.

%{\color{red}
%Equalities may be introduced in the following way:
%$f_i\ge0$ and $-f_i\ge0$ are both in $S_f$.
%}

Let $S_f=\{f_i\geq0,i\in\mathcal{N}_m\}$ be a given inequality set, where $f_i$ is a linear function in $\textbf{x}$.
The following algorithm, called the Implied Equalities Algorithm, finds all the implied equalities of $S_f$.

\begin{algorithm}[H]
\caption{Implied Equalities Algorithm}
\label{Implied-equality-algorithm}
\begin{algorithmic}[1]
\REQUIRE $S_f$.
\ENSURE The implied equalities in $S_f$.
%{\color{red}
\STATE Let $E_0:=\sum_{i=1}^{m}{v_if_i}$, where
$V=\{v_i,i\in \mathcal{N}_m\}$ is  a set of variables. 
%and $W=\{w_j=0,\ j\in \mathcal{N}_n\}$ is a linear system in $V$.
% and denote
%$W = \{w_1=0,\ldots,w_n=0\}.

\STATE Set $E_0\equiv\sum_{j=1}^{n}w_jx_j\equiv0$. Then $W=\{w_j=0,\ j\in \mathcal{N}_n\}$ is a linear system in $V$.

\STATE Solve the linear equations $\{w_j=0,\ j\in \mathcal{N}_n\}$ by Gauss-Jordon elimination to obtain the solution set of $v_i$ of the form $\{v_i=V_i,i\in \mathcal{N}_m\}$, where $d$ is the rank of the linear system $W$ and $V_i$ is a linear function in $m-d$ variables of $V$.

\STATE For every $k\in\mathcal{N}_m$, let $L_k,k=1,\ldots,m$ be the following linear programming problem:
\begin{equation}\label{LP}
\begin{array}{ll}
\ \ \ \ {\rm max}(V_k)\\
{\rm s.t.}\ V_i\geq0,\ i=1,2,\ldots,m.
\end{array}
\end{equation}
\STATE The equality $f_k=0$ is an implied equality of $S_f$ if and only if the optimal value of $L_k$ max$(V_k)>0$.

\RETURN All implied equalities $f_k$'s in $S_f$.
\end{algorithmic}
\end{algorithm}

%Solve the linear inequalities $I_v=\{V_i\geq0, i\in \mathcal{N}_m \}$ to find the possible positive solution of $v_i$. Actually, we have two ways to solve this problem.
%
%The first one is to solve the semi-algebraic system $I_v$ by symbolic computation (such as the "RealTriangularize" package in Maple).
%
%The second way is to use linear programming. We transform the problem to $m$ LP problems.

With Algorithm \ref{Implied-equality-algorithm}, we can obtain the set of implied equalities of $S_f$, denoted by $E_{\tilde{f}}$. 
%According to Proposition \ref{rsb-def} and Definition \ref{reducedstandardbasis}, $\widetilde{\mathcal{R}}^{-1}(\widetilde{E})$ is the reduced standard basis of $\widetilde{\mathcal{R}}^{-1}({E}_{\tilde{f}})$, and $R_f$ is obtained by reducing $\mathcal{R}^{-1}(S_f)$ by $\widetilde{\mathcal{R}}^{-1}(\widetilde{E})$. 
%
%
%From the Algorithm \ref{Implied-equality-algorithm}, we obtain the following theorem.
%\begin{theorem}\label{II.2}
%There is no implied equalities in $\mathcal{R}(R_f)$, i.e. $\mathcal{R}(R_f)$ is a pure inequality set.
%\end{theorem}
The following example illustrates how we can apply Algorithm~\ref{Implied-equality-algorithm} and then Algorithm~\ref{division-algorithm}
to reduce a given inequality set.
A justification of Algorithm~\ref{Implied-equality-algorithm} is given after the example.

%For explanation, we give an example.
\begin{example} 
Fix the variable order $x_1\prec x_2\prec x_3$. Let $S_f=\{f_1\geq0,f_2\geq0,f_3\geq0,f_4\geq0,f_5\geq0\}$, where $f_1=x_1,\ f_2=x_2-x_1,\ f_3=-x_1,\ f_4=-x_2$ and $f_5=x_2+x_3$.
An application of Algorithm \ref{Implied-equality-algorithm} to $S_f$ yields the following:
\begin{itemize}

\item
Firstly, we let $E_0=\sum_{i=1}^{5}v_if_i=\sum_{j=1}^{3}w_jx_j$. Then we have $V=\{v_1,v_2,v_3,v_4,v_5\}$ and $W=\{w_1=0,w_2=0,w_3=0\}$ with $w_1=v_1-v_2-v_3$, $w_2=v_2-v_4+v_5$ and $w_3=v_5$.

\item
The rank of $W$ is $d=3$. We then solve the linear equations $W$ by Gauss-Jordon elimination
to obtain $\{v_i=V_i,i\in\mathcal{N}_5\}$, where $V_1=v_3+v_4$, $V_2=v_4$, $V_3=v_3$, $V_4=v_4$ and $V_5=0$, from which we can see that $V_i$ is a linear function of the two variables $v_3$ and $v_4$.

\item
Finally, we have the following $5$ linear programming problems:

$L_1:\ \ {\rm max}(v_3+v_4)\ \ \ {\rm s.t.}\ \ v_3+v_4\ge0,\ v_3\ge0,\ v_4\ge0.$

$L_2:\ \ {\rm max}(v_4)\ \ \ \ \ \ \ \ \ {\rm s.t.}\ \ v_3+v_4\ge0,\ v_3\ge0,\ v_4\ge0.$

$L_3:\ \ {\rm max}(v_3)\ \ \ \ \ \ \ \ \ {\rm s.t.}\ \ v_3+v_4\ge0,\ v_3\ge0,\ v_4\ge0.$

$L_4:\ \ {\rm max}(v_4)\ \ \ \ \ \ \ \ \ {\rm s.t.}\ \ v_3+v_4\ge0,\ v_3\ge0,\ v_4\ge0.$

$L_5:\ \ {\rm max}(0)\ \ \ \ \ \ \ \ \ \ {\rm s.t.}\ \ v_3+v_4\ge0,\ v_3\ge0,\ v_4\ge0.$

\item
Observe that $L_2$ and $L_4$ are same, and the optimal value of $L_5$ is $0$. Then, we solve $L_1$ to $L_3$ to obtain that the optimal values are all equal to $+\infty$. 
%and both the optimal values of $L_2$ and $L_3$ are also $+\infty$. Thus the optimal values of $L_1,L_2,L_3,L_4$ are $+\infty$, and 
%The optimal value of $L_5$ is $0$.
%
Thus, we obtain the implied equality set, denoted by $E_{\tilde{f}}=\{\tilde{f}_1=0,\tilde{f}_2=0,\tilde{f}_3=0,\tilde{f}_4=0\}$, where $\tilde{f}_1=x_1$, $\tilde{f}_2=x_2-x_1$, $\tilde{f}_3=-x_1$ and $\tilde{f}_4=-x_2$.

\end{itemize}

Upon applying Algorithm~\ref{Implied-equality-algorithm}, the inequality set 
$S_f$ is reduced to the inequality set $S_f^\prime = \{ f_5 \ge 0 \} = \{ x_2+x_3 \ge 0 \}$
constrainted by the equality set $E_{\tilde{f}}$. Finally, apply 
Algorithm~\ref{division-algorithm} with $S_f^\prime$ and $E_{\tilde{f}}$ as inputs to obtain 
$R_f = \{ x_3 \}$.
In other words, the inequality set $S_f$ is reduced to $\{ x_3 \ge 0 \}$ constrained by 
the equality set $E_{\tilde{f}}$ after the applications of Algorithm~\ref{Implied-equality-algorithm} and then 
Algorithm~\ref{division-algorithm}. 

%By Proposition \ref{rsb-def}, we obtain $\widetilde{E}=\{x_1=0,x_2=0\}$ from $E_{\tilde{f}}$. Thus we obtain the reduced standard basis $\widetilde{\mathcal{R}}^{-1}(\widetilde{E})=\{x_1,x_2\}$.
%Then we get the remainder set $R_f=\{x_3\}$ by the Algorithm \ref{division-algorithm}.

\end{example}

\noindent
\textbf{Justification for Algorithm \ref{Implied-equality-algorithm}}. In Algorithm \ref{Implied-equality-algorithm},
the optimal value of $L_k$ being positive means that we can find a set of values of $v_i,\ i\in\mathcal{N}_m$ satisfying $v_k>0$ and $v_j\geq0$ for $j\neq k$, such that $\sum_{i=1}^{m}{v_if_i}\equiv0$, which can be rewritten as 
\[
f_k\equiv\sum_{i=1,i\neq k}^{m}{\left(-\dfrac{v_i}{v_k}\right)f_i}. 
\]
Since by Lemma~\ref{equa1}, $f_k=0$ is an implied equality if and only if $f_{k}\equiv\sum\limits_{i=1,i\neq k}^{m}p_{i}f_{i}$ with $p_{i}\leq0$ for $i\in\mathcal{N}_m$, 
we see that the equality $f_k=0$ is an implied equality of $S_f$ if and only if the optimal value of $L_k$ is positive.

\subsection{Minimal characterization set}

In this subsection, we first define a minimal characterization set of an inequality set and
prove its uniqueness. Then we present an algorithm to obtain this set.

\begin{define}\label{mcs}
Let ${S}_g=\{g_i\geq0,i\in\mathcal{N}_m\}$ be an inequality set and ${S}_{g'}=\{g'_i\geq0,i\in\mathcal{N}_{m'}\}$ be a subset of $S_{g}$. If
%is called the minimal characterization set of $S_f$ if the following two conditions holds:

1) ${S}_{g}$ and ${S}_{g'}$ are equivalent, and

%2) there is no implied equalities in $S_{f'}$

2) there is no redundant inequalities in ${S}_{g'}$,
\\
we say that $S_{g'}$ is a minimal characterization set of $S_{g}$.
\end{define}

\begin{define}
Let $S_g=\{g_i\geq0,i\in\mathcal{N}_m\}$ and $S_{g'}=\{g'_i\geq0,i\in\mathcal{N}_{m'}\}$ be two inequality sets. We say
$P_{g'}=\mathcal{R}^{-1}(S_{g'})$ is a minimal characterization set of $P_g=\mathcal{R}^{-1}(S_g)$ if $S_{g'}$ is a minimal characterization set of $S_g$.
\end{define}

\begin{proposition}\label{PropII.2}
Let $S_g=\{g_i\geq0,i\in\mathcal{N}_m\}$ be an inequality set. If $S_{g'}=\{g'_i\geq0,i\in\mathcal{N}_{m'}\}$ is a minimal characterization set of $S_{g}$, then $m'\leq m$ and $0\notin \mathcal{R}^{-1}({S}_{g'})$.
%and $S_{g'}$ is also a pure inequality set.
\end{proposition}

\begin{proof}
%If $S_{g'}$ is a minimal characterization set of $S_{g}$, then
%by the Definition \ref{mcs}, we have $S_{g'}\subseteq S_{g}$, which means $\mathcal{R}^{-1}(S_{g'})\subseteq\mathcal{R}^{-1}(S_{g})$ by Definition \ref{subset}. Thus, $m'\le m$.
Since $S_{g'} \subseteq S_q$ by Definition \ref{mcs}, we have $m'\le m$.
In addition, if $0\in \mathcal{R}^{-1}({S}_{g'})$, then $0\ge0$ is a redundant inequality in $S_{g'}$, which contradicts that $S_{g'}$ is a minimal characterization set of $S_{g}$. Thus, $0\notin \mathcal{R}^{-1}({S}_{g'})$.
\end{proof}

%Note that the minimal characterization set of a pure inequality set is also a pure inequality set.

The following corollary is immediate from Definition \ref{mcs} and Proposition \ref{subsetofpureispure}.

\begin{corollary}\label{CorIII.1}
A minimal characterization set of a pure inequality set is also a pure inequality set.
\end{corollary}

\begin{theorem}\label{unique}
Let $h_1,\ldots,h_m\in\mathbb{R}_h[\mathbf{x}]$ and
$S_h=\{h_i\geq0,i\in\mathcal{N}_m\}$ be a pure inequality set.
Then the minimal characterization set of $S_h$ is unique.
\end{theorem}

\begin{proof}
%Let's start with two straightforward observations from Algorithm 1:
%$0\notin \mathcal{R}^{-1}[{S}_h]$, $\mathcal{S}_h\subset S_h$, where $\mathcal{S}_h$ is the minimal characterization set of $S_h$.
%

Consider two minimal characterization sets of a pure set of linear inequalities $S_h$, denoted by $\mathcal{S}_{h'}=\{h'_i\geq0,i\in \mathcal{N}_{m_1}\}$ and $\mathcal{S}_{\bar{h}}=\{\bar{h}_i\geq0,i\in \mathcal{N}_{m_2}\}$. By Definition \ref{mcs}, $S_{h'}$ and $S_{\bar{h}}$ are equivalent, and by Corollary~\ref{CorIII.1},
they are both pure inequality sets.
We will prove by contradiction that $\mathcal{S}_{h'}$ and $\mathcal{S}_{\bar{h}}$ are trivially equivalent.

Assume that for some inequality $(h'_{j}\geq0)\in S_{h'}$, we cannot find $(\bar{h}_i\geq0)\in S_{\bar{h}}$ that is trivially equivalent to $h'_{j}\geq0$.
By Proposition \ref{Prop2} and Lemma \ref{GG}, we have 
\[
h'_j\equiv\sum_{i=1}^{m_2}p_i\bar{h}_i,
\]
with $p_i\geq0$.
Without loss of generality, assume that
 $p_{i}>0$ for $i=1,\ldots,\bar{m}_2$ and $p_{i}=0$ for $i=\bar{m}_2+1,\ldots,m_2$, where $2\leq\bar{m}_2\leq m_2$.
%where at least two $p_i$'s are positive and other $p_i$'s are equal to $0$.
%
Again by Lemma \ref{GG}, for all $i \in \mathcal{N}_{m_2}$,
\begin{eqnarray}\label{barh1}
\bar{h}_i\equiv\sum\limits_{k=1}^{m_1}q_{i,k}h'_{k},
\end{eqnarray}
where $q_{i,k}\geq0$. Then
\begin{eqnarray}\label{f_j}
 h'_j\equiv\sum_{i=1}^{\bar{m}_2}p_i\bar{h}_i
%+\sum_{i=\bar{m}_2}^{m_2}p_i\bar{h}_i
 \equiv\sum_{i=1}^{\bar{m}_2}p_i\sum_{k=1}^{m_1}q_{i,k}h'_k.
\end{eqnarray}

Rewrite \eqref{f_j} as
\begin{eqnarray}\label{a3}
\left(1-\sum\limits_{i=1}^{\bar{m}_2}p_iq_{i,j}\right)h'_j(\mathbf{x}) \equiv
\sum\limits_{i=1}^{\bar{m}_2}p_i\sum\limits_{k\in\mathcal{N}_{m_1}\backslash\{j\}}q_{i,k}h'_{k}(\mathbf{x}).
%=\sum\limits_{i=1}^{\bar{m}_2}p_i\sum\limits_{k\neq j}^{\bar{m}_{i,1}}q_{i,k}h'_{k}(\mathbf{x}).
\end{eqnarray}
By collecting the coefficients of $h'_k(\mathbf{x})$ on the RHS, we have
\begin{eqnarray}\label{q9wqfq}
\left(1-\sum\limits_{i=1}^{\bar{m}_2}p_iq_{i,j}\right) h'_j(\mathbf{x}) \equiv
\sum\limits_{k\in\mathcal{N}_{m_1}\backslash\{j\}}  a_k h'_{k}(\mathbf{x}).
%=\sum\limits_{i=1}^{\bar{m}_2}p_i\sum\limits_{k\neq j}^{\bar{m}_{i,1}}q_{i,k}h'_{k}(\mathbf{x}),
\end{eqnarray}
where 
\begin{eqnarray}\label{2finup}
a_k=\sum\limits_{i=1}^{\bar{m}_2}p_iq_{i,k}.
\end{eqnarray}

Now in \eqref{barh1}, for a fixed $i\in\mathcal{N}_{m_2}$,
if $q_{i,k}=0$ holds for all $k=1,\ldots,m_1$ such that $k\neq j$, then we have
\begin{eqnarray}\label{barh}
\bar{h}_i\equiv\sum_{k=1}^{m_1}q_{i,k}h'_k\equiv q_{i,j}h'_j.
\end{eqnarray}
If $q_{i,j}>0$, then $\bar{h}_i$ and $h'_{j}$ are trivially equivalent, contradicting our assumption that there exists no $\bar{h}_i\in \mathcal{S}_{\bar{h}}$ which is trivially equivalent to $h'_{j}$. On the other hand, if $q_{i,j}=0$, then $\bar{h}_i\equiv 0$, which by Proposition \ref{PropII.2} contradicts the assumption that $\mathcal{S}_{\bar{h}}$ is a minimal characterization set of $S_h$.
%which constricts our assumption if $q_{i,j}>0$, and constricts $0\notin \mathcal{R}^{-1}(S_{\bar{h}})$ if $q_{i,j}=0$.
%then \eqref{f_j} becomes
%$$h'_j=\sum_{i=1}^{\bar{m}_2}p_i q_{i,j}h'_j,$$
%which means $h'_j=0$, constricting to $0\notin \mathcal{R}^{-1}[S_{h'}]$.
Thus we conclude that for every $i\in\mathcal{N}_{m_1}$, there exists at least one $k\in\mathcal{N}_{m_1}\backslash\{j\}$ such that $q_{i,k}>0$.
%Without loss of generality, we set $q_{i,k}>0$ for $k=1,\ldots,\bar{m}_{i,1}$, and $q_{i,k}=0$ for $k=\bar{m}_{i,1},\ldots,{m}_1$,
%with $1\leq\bar{m}_{i,1}\leq m_1$.
From this and (\ref{2finup}), it is not difficult to see that on the RHS of (\ref{q9wqfq}), there exists
at least one $k \in\mathcal{N}_{m_1}\backslash\{j\}$ such that $a_k > 0$.

Consider a solution $\mathbf{x}^*$ of $S_{h'}$ such that $h'_{k}(\mathbf{x}^*) > 0$ for 
all $k \in \mathcal{N}_{m_1}$.
Such an $\mathbf{x}^*$ exists
because $S_{h'}$ is a pure inequality set.
Substituting $\mathbf{x}=\mathbf{x}^*$ in \eqref{q9wqfq} to yield
\begin{eqnarray}\label{a4}
\left(1-\sum\limits_{i=1}^{\bar{m}_2}p_iq_{i,j}\right) h'_j(\mathbf{x}^*) =
\sum\limits_{k\in\mathcal{N}_{m_1}\backslash\{j\}}  a_k h'_{k}(\mathbf{x}^*).
\end{eqnarray}
Since there exists
at least one $k \in\mathcal{N}_{m_1}\backslash\{j\}$ such that $a_k > 0$, the RHS 
above is strictly positive, which implies that $1-\sum\limits_{i=1}^{\bar{m}_2}p_iq_{i,j}>0$.
It then follows that $h'_{j}$ can be written as a conic combination of $h'_{k}$, $k\in\mathcal{N}_{m_1}\backslash\{j\}$. 
In other words, $h'_{j} \ge 0$ is implied by $h'_{k} \ge 0$, $k\in\mathcal{N}_{m_1}\backslash\{j\}$.
This contradicts that $\mathcal{S}_{h'}$ is a minimal characterization set of $S_h$.

%which implies $1-\sum\limits_{i=1}^{\bar{m}_2}p_iq_{i,j}\geq0$. If $1-\sum\limits_{i=1}^{\bar{m}_2}p_iq_{i,j}>0$, since there exists
%at least one $k \in\mathcal{N}_{m_1}\backslash\{j\}$ such that $a_k > 0$, it follows from \eqref{q9wqfq} that $h'_{j}$ can be written as a conic combination of $h'_{k}$, $k\in\mathcal{N}_{m_1}\backslash\{j\}$. 
%In other words, $h'_{j} \ge 0$ is implied by $h'_{k} \ge 0$, $k\in\mathcal{N}_{m_1}\backslash\{j\}$.
%This contradicts that $\mathcal{S}_{h'}$ is a minimal characterization set of $S_h$. 
%
%If $1-\sum\limits_{i=1}^{\bar{m}_2}p_iq_{i,j}=0$, then \eqref{q9wqfq} becomes
%\begin{eqnarray}\label{a5}
%\sum\limits_{k\in\mathcal{N}_{m_1}\backslash\{j\}}a_kh'_{k}(\mathbf{x})
%\equiv0.
%\end{eqnarray}
%If there exists only one $k\in\mathcal{N}_{m_1}\backslash\{j\}$ such that $a_k>0$, then $h'_{k}(\mathbf{x})\equiv0$, which by Proposition \ref{PropII.2} contradicts the assumption that $\mathcal{S}_{\bar{h}}$ is a minimal characterization set of $S_h$.
%On the other hand, if there exist at least two $k\in\mathcal{N}_{m_1}\backslash\{j\}$ such that $a_k>0$, then by Definition \ref{def-purein}, it contradicts that $S_{h'}$ is a pure inequality set.
%%
%%From \eqref{f_j}, we know $h'_j$ is a positive linear combination of some polynomials in $\{h'_{k}|k\neq j\}$, which contradicts that $\mathcal{S}_{h'}$ is the minimal characterization set.

Summarizing the above, we have proved that for every $(h'_{j} \ge 0) \in S_{h'}$, we can find an $(\bar{h}_i \ge 0) \in S_{\bar{h}}$ which is trivially equivalent to $h'_{j} \ge 0$. Moreover, $\bar{h}_i$ is unique, which can be seen as
follows. If there exists another $(\bar{h}_{i'} \ge 0) \in S_{\bar{h}}$ which is trivially equivalent to $h'_{j} \ge 0$, then $\bar{h}_i \ge 0$ and $\bar{h}_{i'} \ge 0$ are also trivially equivalent to each other, contradicting that $S_{h'}$ is a minimal characterization set of $S_{h}$. 
In the same way, we can prove that for every $(\bar{h}_{i} \ge 0) \in S_{\bar{h}}$, we can find 
a unique $(h'_j \ge 0) \in S_{h'}$ which is trivially equivalent to $\bar{h}_{i} \ge 0$. Thus, $S_{h'}$ and $S_{\bar{h}}$ are trivially equivalent and have exactly the same number of 
inequalities, which means that the minimal characterization set of a pure inequality set $S_h$ is unique. This completes the proof of the theorem.
\end{proof}

%Similar to the proof of Theorem \ref{unique}, we can obtain the following theorem.
\begin{theorem}\label{V.3}
Let $S_f=\{f_i\geq0,i\in\mathcal{N}_{m_1}\}$ and $S_{g}=\{g_{i},i\in\mathcal{N}_{m_2}\}$
be two pure inequality sets, and
$S_{f'}$ and $S_{g'}$ be their minimal characterization sets respectively. If $S_f$ and $S_g$ are equivalent, then $S_{f'}$ and $S_{g'}$ are trivially equivalent.
\end{theorem}

\begin{proof}
If the two pure inequality sets $S_f$ and $S_g$ are equivalent, then $S_{f'}$ and $S_{g'}$ are pure and equivalent. Thus the theorem follows immediately from the proof of Theorem \ref{unique}.
\end{proof}

Next, we give an example to show that the minimal characterization set of a non-pure inequality set may not be unique.

\begin{example}
Let $S_f=\{f_1\geq0,f_2\geq0,f_3\geq0,f_4\geq0,\}$ be an inequality set, where $f_1=x_1-x_2$, $f_2=x_2$, $f_3=-x_2$, $f_4=x_1$. Evidently, $S_f$ is a non-pure inequality set, and it 
can readily be seen that both $S_{f'}=\{f_1\geq0,f_2\geq0,f_3\geq0\}$ and $S_{f''}=\{f_2\geq0,f_3\geq0,f_4\geq0\}$ are minimal characterization sets of $S_f$. However, $S_{f'}$ and $S_{f''}$ are not trivially equivalent. 
Thus, the minimal characterization set of $S_f$ isn't unique.
\end{example}

%
%Rewriting $Ax\leq b$ and $A'x\leq b'$ as polynomials' forms $\{f_i\geq0,\ i=1,..,m\}$ and $\{f'_i\geq0,\ i=1,..,m'\}$ respectively. The constraint polynomial set $\{f'_i|f'_i\geq0,\ i=1,..,m'\}$ is called the minimal characterization set of $\{f_i|f_i\geq0,\ i=1,..,m\}$ if the two conditions in Definition \ref{mcs} satisfy.

%\begin{lemma}[Farkas' Lemma\cite{Farkas1902,Achiya1997}]
%\label{Farkas}
%Let $\mathbf{A}\in \mathbb{R}^{m\times n}$ and $\mathbf{b}\in\mathbb{R}^{m}$. Then exactly one the following two assertions is true:

%1. There exists an ${\rm \mathbf{x}}\in\mathbb{R}^n$ such that $\mathbf{A}{\rm \mathbf{x}}=\mathbf{b}$ and ${\rm \mathbf{x}}\ge0$.

%2. There exists a ${\rm \mathbf{y}}\in\mathbb{R}^m$ such that $\mathbf{A}^{T}{\rm \mathbf{y}}\geq0$ and $\mathbf{b}^{T}\mathbf{y}<0$.

%Here, the notation ${\displaystyle \mathbf{x} \geq 0}$ means that all components of the vector ${\displaystyle \mathbf{x}}$ are nonnegative.

%\end{lemma}

Let $S_h=\{h_{i}\geq0,\ i\in \mathcal{N}_m\}$ be an inequality set, where $h_i\in \mathbb{R}_h[\mathbf{x}]$.
%is a homogeneous linear function of $\mathbf{x}$. 
Based on Lemma \ref{GG}, the following algorithm, called Minimal Characterization Set Algorithm, can be used to obtain a minimal characterization set of $S_h$.

%%%-------------------------------------------------------------------------------------------------------------------------------------
\begin{algorithm}[H]
\caption{Minimal Characterization Set Algorithm}
\label{minimal}
%\begin{spacing}{2}
\begin{algorithmic}[1]
\REQUIRE 
%A homogeneous linear inequality set $S_h=\{h_i\geq0,i\in\mathcal{N}_{m}\}$.
$S_h$.

\ENSURE 
A minimal characterization set of $S_h$.\\
\rule{15cm}{0.05em}\\
Set $P_h:=\mathcal{R}^{-1}(S_h)$, $\mathcal{M}:=\mathcal{N}_m$.

\FOR{$k$ from 1 to $m$}

%\STATE Select one $h_k\in \mathcal{S}_h$.

\STATE Let $H_k:=h_k-\sum\limits_{i\in\mathcal{M}\backslash\{k\}}q_{i,k}h_i$, where $T_k=\{q_{i,k},i\in\mathcal{M}\backslash\{k\}\}$ is a set of variables.
%where $q_{i,k}\in\mathbb{R}$ for $i\in\mathcal{M},i\neq k$, are unknowns to be determined, and $Q_{i,k}$ is the coefficient of $x_i$ in $H_k$, so $Q_{i,k}$ is a function of $q_{1,k},\ldots,q_{m,k}$.

\STATE Set $H_k\equiv\sum\limits_{i=1}^{n}Q_{i,k}x_i\equiv 0$. Then $\widetilde{T}_k=\{Q_{i,k}=0,i\in\mathcal{N}_n\}$ is a linear system in $T_k$.

%\STATE Setting $H_k=0$ yields $Q_{i,k}=0,\ i=1,\ldots,n$.

\STATE Solve the linear equations of $\widetilde{T}_{k}$.

%$q_{i,k}$ from $\{Q_{i,k}=0,\ i=1,\ldots,n\}$ results in the solutions $\{q_{i,k}=\mathcal{Q}_{i,k},\ i=1,\ldots,n\}$, where $\mathcal{Q}_{i,k}$ is a linear combination of $q_{j,k}$, for $j>i$.

\IF{the linear equations of $\widetilde{T}_k$ can be solved}

\STATE Obtain the solution set of $q_{i,k}$ of the form $\{q_{i,k}=\mathcal{Q}_{i,k},i\in\mathcal{M}\backslash\{k\}\}$, where $d_1$ is the rank of the linear system $\widetilde{T}_k$ and $\mathcal{Q}_{i,k}$ is a linear function in $N[\mathcal{M}\backslash\{k\}]-d_1$ variables of $T_k$.
%\IF{there exist solutions of $q_{j,k}$'s by solving $\{Q_{i,k}=0,\ i=1,\ldots,n\}$}

%\STATE Solving $\{\mathcal{Q}_{i,k}\geq0,\ i=1,\ldots,n\}$ gives the solutions set of $q_{i,k}$.
\STATE Let $L_{k}$ be the following linear programming problem:
$$\begin{array}{ll}
&{\rm min}(0)\\
{\rm s.t.}\ &\mathcal{Q}_{i,k}\ge0,\ i\in\mathcal{M}\backslash\{k\}.
\end{array}$$

%\IF{there exist positive $q_{j,k}$'s in the solutions set of $\{\mathcal{Q}_{i,k}\geq0,\ i=1,\ldots,n\}$}
%\IF{there exist solutions of $\{\mathcal{Q}_{i,k}\geq0,\ i=1,\ldots,n\}$}
\IF{$L_k$ can be solved}

\STATE $P_h:=P_h\backslash \{h_k\}$, $\mathcal{M}:=\mathcal{M}\backslash \{k\}$.

\ENDIF
\ENDIF
\ENDFOR

\RETURN $\mathcal{R}(P_h)$.
\end{algorithmic}
%\end{spacing}
\end{algorithm}

%Note that the set $\mathcal{R}[P_h]$ obtained by Algorithm 1 is the minimal characterization set of $S_h$.

\noindent
\textbf{Justification for Algorithm \ref{minimal}}.
Steps 2 to 11 remove the polynomial $h_k$ from $P_h$ if 
it can be expressed as a conic combination of $h_i, i\in\mathcal{M}\backslash\{k\}$.
Iterating over all $k$ from 1 to $m$, the output inequality set $\mathcal{R}(P_h)$
is equivalent to $S_h$ and it is a pure 
inequality set. Hence, it is a minimal characterization set of $S_h$.

%\begin{theorem}
%The set $\mathcal{R}(P_h)$ obtained by Algorithm \ref{minimal} is a minimal characterization set of $S_h$.
%%The minimal characterization set of a system of homogeneous linear inequalities can be found by using Algorithm \ref{minimal}.
%\end{theorem}
%\begin{proof}
%For the Steps 2-11 in Algorithm \ref{minimal}, we remove the polynomial which is a conic combination of some other polynomials in $\mathcal{R}^{-1}(S_h)$. By adding a ``for" loop (Steps 1 and 12), we remove all redundant inequalities from $\mathcal{R}^{-1}(S_h)$. 
%Thus the final output set $\mathcal{R}(P_h)$ is a subset of $S_h$, and equivalent to $S_h$. Since $\mathcal{R}(P_h)$ has no redundant inequalities, by Definition \ref{mcs}, the final $\mathcal{R}(P_h)$ is a minimal characterization set of $S_h$.
%\end{proof}

%Note that by Theorem \ref{unique}, if $S_h$ is a homogeneous linear pure inequality set, then the minimal characterization set obtained by Algorithm \ref{minimal} is unique.
%By Definition \ref{def-rmcs}, we can use the \textbf{Implied Equalities Algorithm} and Algorithm \ref{minimal} to find the reduced minimal characterization set of an arbitrary inequality set.

\subsection{The reduced minimal characterization set}
In this subsection, we first define the reduced minimal characterization set of a linear inequality set and prove its uniqueness. Then we present an algorithm to obtain this set.

Let $S_f=\{f_i\geq0,i\in\mathcal{N}_m\}$ be a linear inequality set, 
and $E_{\tilde{f}}$ be the set of implied equalities of $S_f$ obtained by applying 
Algorithm~\ref{Implied-equality-algorithm}.
Then we obtain $\tilde{E}$, the Jordan normal form of $E_{\tilde{f}}$, as in Proposition~\ref{rsb-def}.
Let $R_f$ be the remainder set obtained by
reducing $\mathcal{R}^{-1}(S_f) \backslash {\cal R}^{-1}(E_{\tilde{f}})$ by $\widetilde{\mathcal{R}}^{-1}(\widetilde{E})$
using Algorithm~\ref{division-algorithm}.

%{\color{blue}
\begin{theorem}\label{Rf unique}
The set $\mathcal{R}(R_f)$ is a pure inequality set.
\end{theorem}
\begin{proof}
Let $\widetilde{E}=\{E_i=0,i\in \mathcal{N}_{\widetilde{n}}\}$, and assume there is an implied equality $(\bar{f}=0)\in \mathcal{R}(R_f)$. 
%Observe that $\bar{f}$ is obtained by reducing some $f\in\mathcal{R}^{-1}(S_f)$ by $\widetilde{\mathcal{R}}^{-1}(\widetilde{E})$. 
In the process of obtaining $\bar{f}$,
we substitute $x_{k_i} = U_i, i \in \tilde{\cal N}_{\tilde{n}}$ into some polynomial $f\in\mathcal{R}^{-1}(S_f)$ (cf.\ \ref{Sta-Basis}).
Therefore, we can write 
\begin{equation}
\bar{f} \equiv f - \sum\limits_{i=1}^{\tilde{n}}c_i E_i ,
\label{4q9pv}
\end{equation}
where $c_i$ is the coefficient of $x_{k_i}$ in $f$.
Let ${\bf x}^*$ be a solution of $S_f$.
From Remark~\ref{4398hqf}, we see that 
${\bf x}^*$ is also a solution of $\mathcal{R}(R_f)$ constrained by $\widetilde{E}$,
so that $E_i({\bf x}^*) = 0$ for all $i \in \mathcal{N}_{\widetilde{n}}$.
From (\ref{4q9pv}), we have 
\[
f({\bf x}^*)=\bar{f}({\bf x}^*)-\sum\limits_{i=1}^{\tilde{n}}c_i E_i({\bf x}^*) .
\]
Since $\bar{f} = 0$ is an implied inequality of $S_f$, 
we have $\bar{f}({\bf x}^*) = 0$.
It follows from the above that $f({\bf x}^*)=0$. Since this holds for all solution 
${\bf x}^*$ of $S_f$, we see that $f=0$ is an implied equality of $S_f$, i.e., $(f = 0) 
\in E_{\tilde{f}}$, which is a contradiction to $f \in \mathcal{R}^{-1}(S_f) \backslash {\cal R}^{-1}(E_{\tilde{f}})$. The theorem is proved.
\end{proof}
%}

Since $\mathcal{R}(R_f)$ is a pure inequality set, the minimal characterization set of $\mathcal{R}(R_f)$ is unique. We let $S_{r'}$ be the minimal characterization set of $\mathcal{R}(R_f)$.
%then we give the following definition.

\begin{define}\label{def-rmcs}
%Consider an inequality set $S_f=\{f_i\geq0,i\in\mathcal{N}_m\}$. Let the set of implied equalities obtained from $S_f$ be $S_{\tilde{f}}=\{\tilde{f}_i=0,i\in\mathcal{N}_{\tilde{m}}\}$, the reduced standard basis of $\widetilde{\mathcal{R}}^{-1}(S_{\tilde{f}})$ be $G$, and the remainder set of division on $\mathcal{R}^{-1}(S_f)$ by $G$ be $R_s$. Then the set $S_M=\widetilde{\mathcal{R}}(G)\cup \mathcal{R}(R_s)$ is called the reduced minimal characterization set of $S_f$.
% 
%Using the notations introduced above,
The set $S_M=\widetilde{E}\cup S_{r'}$ is called the reduced minimal characterization set of $S_f$.
\end{define}

\begin{theorem}\label{reducemcs-unique}
The reduced minimal characterization set of $S_f$ is unique.
\end{theorem}

\begin{proof}
Fix the variable order $x_1\prec x_2\prec \cdots\prec x_n$. By Proposition \ref{rsb-def}, the reduced standard basis $\widetilde{\mathcal{R}}^{-1}(\widetilde{E})$ is unique, which yields that the remainder set $R_f$ is unique. Since $\mathcal{R}({R_f})$ is a pure inequality set
by Theorem~\ref{unique}, the minimal characterization set of $\mathcal{R}({R_f})$ is unique.
Hence, $S_M$ is unique.
\end{proof}

In the following, we present an algorithm to find the reduced minimal characterization set of a linear inequality set.

\begin{algorithm}[H]
\caption{Reduced Minimal Characterization Set Algorithm}
\label{reducemcs-of-S_f}
\begin{algorithmic}[1]
\REQUIRE
$S_f$.
%An inequality set $S_f=\{f_i\geq0,i\in\mathcal{N}_m\}$.
\ENSURE
The reduced minimal characterization set of $S_f$.

\STATE Apply Algorithm \ref{Implied-equality-algorithm} to find the implied equality set of $S_f$, denoted by $E_{\tilde{f}}$.

\STATE 
Apply Algorithm~\ref{division-algorithm} to reduce $\mathcal{R}^{-1}(S_f) \backslash {\widetilde{\cal R}}^{-1}(E_{\tilde{f}})$ by $E_{\tilde{f}}$ to obtain $R_f$. \\

%{\color{red}$S_f
%\backslash E_{\tilde{f}}$  maybe not suitable, because $S_f$ is inequality set, but $E_{\tilde{f}}$ is equality set. We can change to:\\
%Apply Algorithm \ref{division-algorithm} to reduce $\mathcal{R}^{-1}(S_f) \backslash {\cal R}^{-1}(E_{\tilde{f}})$ by $\widetilde{\mathcal{R}}^{-1}(\widetilde{E})$ to obtain $R_f$.}

\STATE Apply Algorithm \ref{minimal} to obtain the minimal characterization set of $\mathcal{R}(R_f)$, denoted by $S_{r'}$.

\RETURN 
%Equality set $\widetilde{E}$;\\
%Minimal characterization set of $\mathcal{R}(R_f)$ be $S_{r'}$;\\
%The reduced minimal characterization set of $S_f$ be 
$S_{M}=\widetilde{E}\cup S_{r'}$.
\end{algorithmic}
\end{algorithm}

By Proposition \ref{rsb-def} and Theorems \ref{V.3} and \ref{reducemcs-unique}, we immediately obtain the following theorem.
\begin{theorem}
For two equivalent inequality sets, their reduced minimal characterization sets are same.
\end{theorem}

%Note that by Theorem \ref{reducemcs-unique}, we know that for two equivalent inequality sets, their reduced minimal characterization sets are same. 
%Thus, if we obtain the reduced minimal characterization set of $S_h$, denoted by $\mathcal{S}_h$, then we can say $\mathcal{S}_h$ is also the reduced minimal characterization set of the equivalent sets of $S_h$.
%
Note that for a pure inequality set, the minimal characterization set is exactly the reduced minimal characterization set.

%%%%%-------------------------------------------------------------------------------------------

\begin{remark}
Since the basic inequalities contain no implied inequality and hence
form a pure inequality set, the elemental inequalities form the minimal characterization set of the basic inequalities. 
In fact, for a fixed number of random variables, Algorithm \ref{reducemcs-of-S_f} can be used to compute the reduced minimal characterization set of the basic inequalities under the constraint of an equality set and possibly an inequality set (used for example, for including some non-Shannon-type inequalities).
\end{remark}

%\subsection{$s$-variables and transformation}
\section{The $s$-Variables}
\label{sec-s-var}

The $I$-Measure \cite{Yeung1991} gives a set-theoretic interpretation of Shannon's information measure. In this section, we first give a brief introduction to the $I$-Measure. The readers are referred to \cite[Chapter 3]{Yeung2008} for a detailed discussion.
Then we introduce the $s$-variables which facilitate the implementation of the algorithms
to be developed in Section~\ref{LI-algo.}.

%We first give a review of the main results regarding the $I$-measure. For detailed discuss of the $I$-Measure, we refer the reader to \cite{Yeung2008}.
Consider random variables $X_i,\ i=1,\ldots,n$ which are jointly distributed, 
and let $\tilde{X}_i$ be a set variable corresponding to the random variable $X_i$.
Define the universal set $\Omega$ to be $\cup_{i=1}^n\tilde{X}_i$ and let $\mathcal{F}_{n}$ be the $\sigma$-field generated by $W=\{\tilde{X}_i,i=1,\ldots,n\}$. The atoms of $\mathcal{F}_{n}$ have the form $\cap_{i=1}^nY_i$, where $Y_i$ is either $\tilde{X}_i$ or $\tilde{X}_i^{c}$.
Let $\mathcal{A}_{n}\subset \mathcal{F}_{n}$ be the set of all atoms of $\mathcal{F}_{n}$ except for $\cap_{i=1}^n\tilde{X}_i^c$, which is $\emptyset$, the empty set.
Note that $|\mathcal{A}_n|=2^n-1$.
To simplify notations, we shall use $X_G$ to denote $(X_i,i\in G)$, and $\tilde{X}_{G}$ to denote $\cup_{i\in G}\tilde{X}_i$.
%
%We assume throughout that the information problem has $n$ random variables $X_1,X_2, . . . , X_n$.

The $I$-measure $\mu^*$, which is a signed measure on $\mathcal{F}_n$, is constructed by defining $\mu^*{(\tilde{X}_G)}=H(X_G)$ for all nonempty subsets $G$ of $\mathcal{N}_n$. It is consistent with all Shannon's information measure in the sense that the following holds for all (not necessarily disjoint) subsets $G,G',G''$ of $\mathcal{N}_n$ where $G$ and $G'$ are nonempty:
$$\mu^*(\tilde{X}_G\cap\tilde{X}_{G'}-\tilde{X}_{G''})=I(X_G;X_{G'}|X_{G''}).$$

To facilitate the discussion in this paper, we introduce the concept of {\it $s$-variables}, which have the form $s_{i_1,i_2,\ldots,i_n}$,  $i_1,i_2,\ldots,i_n\in \mathcal{N}_n$.
For an integer set $S \subset\mathcal{N}_n$, we denote its minimum by {\rm min}$(S)$.

\begin{define}\label{s-v}
Let $\mu^*$ be an unspecified $I$-measure of $\mathcal{F}_n$. Let $A=\cap_{i=1}^nY_i$ be an atom in $\mathcal{A}_n$
and $S$ be the subset of $\mathcal{N}_n$
such that $Y_i=\tilde{X}_i$ for $i\in S$ and $Y_i=\tilde{X}_i^c$ for $i\in S^c=\mathcal{N}_n\backslash S$.
Replace the integers $i\in S^c$ in the sequence $[n]$ by $*$ to yield the sequence $B_*$.
Then replace all the $*$'s in $B_*$ by {\rm min}$(S)$ to yield another sequence $B_s$.
Let $s_{i_1,i_2,\ldots,i_n}=\mu^*(A)$, where $[i_1,i_2,\ldots,i_n]=B_s$. The variable $s_{i_1,i_2,\ldots,i_n}$ is called the $s$-variable associated with the atom $A$.

\end{define}

Note that in the above definition, there is a one-to-one correspondence between the $s$-variable $s_{i_1,i_2,\ldots,i_n}$ and the atom~$A$. On the one hand, the $s$-variable can be obtained from an atom $A$ as described above. On the other hand, we can determine the associated atom $A$ from the $s$-variable $s_{i_1,i_2,\ldots,i_n}$ through $S$, with $A=\cap_{i=1}^nY_i$, where $Y_i=\tilde{X}_i$ for $i\in S$ and $Y_i=\tilde{X}_i^c$ for $i\in S^c$. This is illustrated in the example below.

%\begin{define}
%Given a sequence $B_*=[i_1,i_2,\ldots,i_n],\ i_j\in \mathcal{N}_n,\ j\in \mathcal{N}_n$, in which some $i_j$'s  are $\ast$, the smallest element in $B_*$ is defined as the minimal integer of set $\{i_1,i_2,\ldots,i_n\}\backslash \{\ast\}$.
%\end{define}

%Referred to Definition \ref{s-v}, we give an illustration to explain the relationship of $s$-variables and the $I$-measure of atoms.

\begin{example}\label{*-theory}
Given the atom $A=\tilde{X}_1\cap\tilde{X}_2^c\cap\tilde{X}_3\cap\tilde{X}_4^c$, we have $S = \{1,3\}$ and $S^c=\{2,4\}$, and $B_*=[1,\ast,3,\ast]$. Replace all the $\ast$'s in $B_*$ by the smallest element in $S$ to yield $B_s=[1,1,3,1]$. Then $s_{1,1,3,1}=\mu^*(A)$ is the $s$-variable corresponding to the atom $A$. On the other hand, from $s_{1,1,3,1}$, we can obtain $S = \{1,1,3,1\}=\{1,3\}$, from which $A$ can be determined.
\end{example}

We now introduce some further notations. Let $t=s_{i_1,i_2,...,i_n}$ be an $s$-variable.
The set $L(t)=\{i_1,i_2,...,i_n\}$ is called the {\it subscript set} of $t$.
The sequence  $\mathcal{L}(t)=[i_1,i_2,...,i_n]$ is called the {\it subscript sequence} of $t$.
The number of elements in the subscript set is denoted by $N[L(t)]$, and
the length of the subscript sequence, denoted by $N[\mathcal{L}(t)]$, is equal to $n$.

For splitting an $s$-variable $s_{i_1,i_2,...,i_n}$, we mean adding an element to $L(s_{i_1,i_2,...,i_n})$ and yielding two new $s$-variables $s_{i_1,i_2,...,i_n,i_1}$ and $s_{i_1,i_2,...,i_n,n+1}$. Note that if $s_{i_1,i_2,...,i_n}$ corresponds to an atom $A\in \mathcal{A}_n$, then $s_{i_1,i_2,...,i_n,i_1}$ and $s_{i_1,i_2,...,i_n,n+1}$ correspond to the atoms $A\cap\tilde{X}_{n+1}^c$ and $A\cap\tilde{X}_{n+1}$ in $\mathcal{A}_{n+1}$, respectively.

%
%Define the positive integer set $N_{>0}=\{1,2,\ldots\}$.

\begin{define}\label{ab}
For $S\subset N_{>0}$ and $a,b\in N_{>0}$, we introduce the following shorthand notations:

` $a\vee b\ \in S$ ' means $a$ $\in S$ {\rm or} $b$ $\in S$,

` $a\wedge b\ \in S$ ' means $a$ $\in S$ {\rm and} $b$ $\in S$,

` $a\backslash b\ \in S$ ' means $a\in S$ {\rm and} $b\notin S$,

` $a\vee b\ \notin S$ ' means $a\notin S$ and $b\notin S$.\footnote{Equivalently, $a\vee b\notin S$ means $\sim(a\in S\ {\rm or}\ b\in S)$.}

%where $a$ and $b$ are positive integers.
\end{define}

Based on Definition \ref{ab}, for $c,d\in N_{>0}$, we further have the following:

` $a\wedge (b\vee c)\ \in S$ ' means $a\in S$ and $b\vee c\ \in S$,

` $(a\vee b)\backslash (c\vee d)\ \in S$ ' means $a\vee b\ \in S$ and $c\vee d\ \notin S$,

` $(a\wedge (b\vee c))\backslash d\ \in S$ ' means $a\wedge (b\vee c)\ \in S$ and $d\ \notin S$.

For $n\in N_{>0}$, let $S_{n}$ be the set of $s$-variables of all the atoms in $\mathcal{A}_n$. Note that $S_{n+1}$ can be obtained from $S_n$. We first illustrate the case $n=1$.
First of all, $S_1=\{s_{1}\}$, where $s_1=\mu^*(\tilde{X}_1)$.
Then, we split $s_1$ to obtain $s_{1,1}$ and $s_{1,2}$ in $S_2$, where $s_{1,1}=\mu^*(\tilde{X}_1-\tilde{X}_2)$ and $s_{1,2}=\mu^*(\tilde{X}_1\cap\tilde{X}_2)$. By also including the additional variable $s_{2,2}=\mu^*(\tilde{X}_2-\tilde{X}_1)$, we obtain $S_2=\{s_{1,1}, s_{1,2}, s_{2,2}\}$.

In general, we can obtain $S_{n+1}$ from $S_n$ as follows.
For every $s$-variable $s_{i_1,i_2,...,i_{n}}$ in $S_{n}$, we split $s_{i_1,i_2,...,i_{n}}$ to obtain $s_{i_1,i_2,...,i_{n},i_1}$ and $s_{i_1,i_2,...,i_{n},n+1}$ in $S_{n+1}$. Then we obtain $S_{n+1}$ by including the additional variable $s_{n,n,...,n}$ with $N[\mathcal{L}(s_{n,n,...,n})]=n+1$.

As illustrations of the use of the notations we have introduced, we state the following which can readily be verified:

%1) $H/I$  $\Leftrightarrow$  $s$-variables $A\in S_n$,

1) $H(X_a,X_b)\,{=}\,\sum t$ \ {for $t$ such that} \ $a\vee b\ \in L(t)$,

2) $I(X_a;X_b)\,{=}\,\sum t$ \ {for $t$ such that} \ $a\wedge b\ \in L(t)$,

3) $H(X_a|X_b)\,{=}\,\sum t$ \ {for $t$ such that} \ $a\backslash b\ \in L(t)$,

4) $I(X_a;(X_b,X_c))\,{=}\,\sum t$ \ {for $t$ such that} \ $a\wedge(b\vee c)\ \in L(t)$,

5) $H(X_a,X_b|X_c,X_d)\,{=}\,\sum t$ \ {for $t$ such that} \ $(a\vee b)\backslash(c\vee d)\ \in L(t)$,

6) $I((X_a;(X_b,X_c))|X_d)\,{=}\,\sum t$ \ {for $t$ such that} \ $(a\wedge (b\vee c))\backslash d\ \in L(t)$.

For example, for three random variables $X_1,\ X_2,\ X_3$, we have the following:
\\[0.2cm]
$H(X_1,X_2)\,{=}\,\mu^*(\tilde{X}_1\cup\tilde{X}_2)=\sum\limits_{1\vee 2\ \in L(s_{i,j,k})}s_{i,j,k}=
s_{1,1,1}+s_{1,1,3}+s_{1,2,1}+s_{1,2,3}+s_{2,2,2}+s_{2,2,3}$,
\\[0.2cm]
$I(X_1;X_2)\,{=}\,\mu^*(\tilde{X}_1\cap\tilde{X}_2)=\sum\limits_{1\wedge 2\ \in L(s_{i,j,k})}s_{i,j,k}=s_{1,2,1}+s_{1,2,3}$,
\\[0.2cm]
$H(X_1,X_2|X_3)\,{=}\,\mu^*(\tilde{X}_1\cup\tilde{X}_2-\tilde{X}_3)=\sum\limits_{(1\vee2)\backslash 3\ \in L(s_{i,j,k})}s_{i,j,k}=s_{1,1,1}+s_{1,2,1}+s_{2,2,2}$,
\\[0.2cm]
$I(X_1;X_2|X_3)\,{=}\,\mu^*(\tilde{X}_1\cap\tilde{X}_2-\tilde{X}_3)=\sum\limits_{(1\wedge 2)\backslash 3\ \in L(s_{i,j,k})}s_{i,j,k}=s_{1,2,1}$.
\\[0.2cm]
Using this set of notations, we can express a Shannon's information measure as a linear polynomial in the $s$-variables which are indexed by subscript sequences, so that they can be conveniently represented in a computer implementation.

\begin{define}[s-variable order]
%The $s$-variable order is defined as the size of the subscript set is ordered in descending order.

Let $t_1=s_{i_1,i_2,\ldots,i_n}$ and $t_2=s_{j_1,j_2,\ldots,j_n}$ be two $s$-variables.
We write $t_1\succ t_2$ if one of the following conditions is satisfied:
\begin{enumerate}
\item[1)] $N[L(t_1)]>N[L(t_2)]$, 
\item[2)] $N[L(t_1)]=N[L(t_2)]$, $i_l=j_l$ for $l=1,\ldots,k-1$ and $i_k<j_k$.
\end{enumerate}
\end{define}

\begin{define}
For $n\in N_{>0}$, let $S_{n}$ be the set of $s$-variables. The associated $s$-variable sequence $\mathcal{S}_n$ is obtained by ordering the elements in $S_n$ according to the $s$-variable order.
%of all the atoms in $\mathcal{A}_n$.
\end{define}

For example, the $s$-variable sequence $\mathcal{S}_3$ is
$
[s_{1,2,3}, s_{1,1,3}, s_{1,2,1}, s_{2,2,3}, s_{1,1,1}, s_{2,2,2}, s_{3,3,3}].
$
The $s$-variable order is employed in the computational procedures to be 
discussed in the next section for the convenience of implementation.

%%%------------------------------------------------------------------------------------------------------------------------------------------
\section{Procedures for proving information inequalities and identities}
\label{procedures}

In this section, we present two procedures for proving information inequalities and identities under the constraint of an inequality set and/or an equality set.
They are designed in the spirit of Theorem~\ref{LP-S}.

%{\color{red}
%Explain What is the input and what is the output.}

\subsection{Procedure I: Proving Information Inequalities}
\noindent
\textbf{Input:}\\ 
Objective information inequality: $\bar{F}\ge0$.\\ 
Additional constraints: $\bar{C}_i=0,\ i=1,\ldots,r_1$; $\bar{C}_j\ge0,\ j=r_1+1,\ldots,r_2$.\\
Element information inequalities: $\bar{C}_k\ge0,\ k=r_2+1,\ldots,r_3$.\\
// \ Here, $\bar{F},\ \bar{C}_i,\ \bar{C}_j,$ and $\bar{C}_k$ are linear combination of information measures.

\noindent
\textbf{Output:} A proof of $\bar{F}\ge0$ if feasible.

Step 1. Construct the $s$-variable set $S_n$ and the associated $s$-variable sequence
 $\mathcal{S}_n$.

Step 2. Transform $\bar{F},\ \bar{C}_i,\ \bar{C}_j$ and $\bar{C}_k$ to linear polynomials $F$, $C_i$, $C_j$ and $C_k$ in $S_n$ respectively.

\noindent
// \
We need to solve \\
// \ \textbf{Problem} $\mathbf{P_1}$: Determine whether $F \ge 0$ is implied by 
\[
\begin{array}{ll}
& C_i=0,\ i=1,\ldots,r_1,\\
& C_j\ge0,\ j=r_1+1,\ldots,r_2,\\
& C_k\ge0,\ k=r_2+1,\ldots,r_3.
\end{array}
\]

Step 3. Apply Algorithm 1 to reduce  $\{C_l,l\in\mathcal{N}_{r_3}\backslash\mathcal{N}_{r_1}\}$ by 
$\{C_l = 0,l\in\mathcal{N}_{r_1}\}$ to obtain the Jordan normal form of
\\ 
\indent\indent 
 $\{C_l,l\in\mathcal{N}_{r_1}\}$,
denoted by $B$, and 
the remainder set, denoted by $\mathbf{C}_1=\{g_i,i\in\mathcal{N}_r\}$.

Step 4. Apply Algorithm \ref{reducemcs-of-S_f} to obtain the reduced minimal characterization set of $\mathcal{R}(\mathbf{C}_1)$, denoted by \\
\indent\indent  $S_M=\widetilde{E}\cup S_{r'}$. Write $S_{r'}=\{\mathbb{C}_j\ge0,j\in\mathcal{N}_{t_2}\}$.

Step 5. Let $G=\widetilde{\mathcal{R}}^{-1}(\widetilde{E})\cup B$ and compute the Jordan normal form of $G$, denoted by $\mathcal{B}=\{\mathcal{C}_{i},i\in\mathcal{N}_{t_1}\}$.

\noindent
// \ In the above, the inequality set $\mathcal{R}(\mathbf{C}_1)$ is generated 
by reducing $\{C_l \ge 0,l\in\mathcal{N}_{r_3}\backslash\mathcal{N}_{r_1}\}$
by $\{C_l = 0,l\in\mathcal{N}_{r_1}\}$, and \\
// \ the inequality set 
$S_{r'}$ is generated by further reducing $\mathcal{R}(\mathbf{C}_1)$
by own implied equalities, which is equivalent to $\widetilde{E}$. \\
// \ Therefore,
in $S_{r'}$, only the free variables in the Jordan normal form $\mathcal{B}$ are involved.

Step 6. Reduce $F$ by $\widetilde{R}(\mathcal{B})$ to obtain the remainder $F_1$. 

\noindent
// \ In both $F_1$ and $S_{r'}$, only the free variables in the Jordan normal form $\mathcal{B}$ are involved. \\
// \ The original Problem ${\rm P}_1$ is now transformed into \\
\noindent
// \ \textbf{Problem $\mathbf{P_2}$}: Determine whether ${F}_1 \ge 0$ is implied by the inequalities in $S_{r'}$, i.e.,
\[
\begin{array}{ll}
& \mathbb{C}_i\ge0,\ j=1,\ldots,t_2.
\end{array}
\]
\noindent

\noindent
// \ Since the equality set $\widetilde{R}(\mathcal{B})$ contains only constraints on the 
pivot variables in $\mathcal{B}$, it is ignored in formulation of \\
// \ Problem ${\rm P}_2$. 
The remaining steps follow Algorithm~\ref{minimal}.

Step 7. Let $x_{j},j\in\mathcal{N}_{n_1}$ be the variables in Problem ${\rm P}_2$. Let $F_2=F_1-\sum_{i=1}^{t_2}p_i\mathbb{C}_i$, where \\
\indent\indent $P= \{p_i,i\in\mathcal{N}_{t_2}\}$
is a set of variables. Set $F_2\equiv\sum_{j=1}^{n_1}q_jx_j\equiv0$. Then $Q=\{q_j=0,j\in\mathcal{N}_{n_1}\}$ \\
\indent\indent is a linear system in $P$.

Step 8. If the linear system $Q$ has no solution, declare that the objective
information inequality $\bar{F}\ge0$ \\
\indent\indent is `Not Provable' and terminate the procedure. 

Step 9. Otherwise, solve the linear equations $\{q_j=0,j\in\mathcal{N}_{n_1}\}$ by Gauss-Jordan elmination to obtain \\
\indent\indent the solution set of $p_i$ in the form $\{p_i=P_i,i\in\mathcal{N}_{t_2}\}$, 
where $P_i$ is
a linear function in $t_2-d_2$ variables \\
\indent\indent of $P$ and $d_2$ is the rank of the linear system $Q$.

Step 10. If $P_i \in \mathbb{R}_{<0}$ (the set of negative real numbers) for some $i\in\mathcal{N}_{t_2}$, declare `Not Provable'.

Step 11. Otherwise, let $S_P$ be the set $\{P_i,i\in\mathcal{N}_{t_2}\}$, and let $\bar{S}_P=S_P\backslash \mathbb{R}$. Write $\bar{S}_P=\{\bar{P}_i,\ i\in\mathcal{N}_{t_3}\}$. \\
\indent\indent If $\bar{S}_P$ is empty, the objective information inequality $\bar{F}$ is proved. Otherwise go to Step 12.

Step 12. \textbf{Problem $\mathbf{P_3}$}:
\[
\begin{array}{ll}
&{\rm min}(0)\\
{\rm s.t.}
& \bar{P}_{i}\ge0,\ i=1,\ldots,t_3 .
\end{array}
\]
\indent\indent If the above LP has a solution, the objective information inequality $\bar{F} \ge 0$ is proved. Otherwise, \\
\indent\indent declare `Not Provable'.

\begin{remark}
Let $N_v(P_1)$, $N_v(P_2)$ and $N_v(P_3)$ be the number of variables in Problems $P_1$, $P_2$ and $P_3$ respectively. Let $N_c(P_1)$, $N_c(P_2)$ and $N_c(P_3)$ be the number of constraints in Problems $P_1$, $P_2$ and $P_3$ respectively.  It is clear that $N_v(P_1)\ge N_v(P_2) \ge N_v(P_3)$, and $N_c(P_1)\ge N_c(P_2) \ge N_c(P_3)$. The reduction of the number of variables and the number of constraints is in general significant. Since most of the computation in the procedure is attributed to solving the LP in Problem P3, compared with the approach in Theorem~\ref{LP-S} where a much larger LP needs to be solved, the efficiency can be significantly improved.
Example~\ref{example-1} illustrates this point.
\end{remark}

\subsection{Procedure II: Proving Information Identities}

\noindent
\textbf{Input:} \\
Objective information identity: $\bar{F}=0$.\\
Additional constraints: $\bar{C}_i=0,\ i=1,\ldots,r_1$; $\bar{C}_j\ge0,\ j=r_1+1,\ldots,r_2$.\\
Element information inequalities: $\bar{C}_k\ge0,\ k=r_2+1,\ldots,r_3$.\\
Here, $\bar{F},\ \bar{C}_i,\ \bar{C}_j,$ and $\bar{C}_k$ are linear combination of information measures.

\noindent
\textbf{Output:}
A proof of $\bar{F}=0$ if feasible.

Step 1. Construct the $s$-variable set $S_n$ and the associated $s$-variable sequence
 $\mathcal{S}_n$.

Step 2. Transform $\bar{F},\ \bar{C}_i,\ \bar{C}_j$ and $\bar{C}_k$ to linear polynomials $F$, $C_i$, $C_j$ and $C_k$ in $S_n$ respectively.

\noindent
// \
We need to solve \\
// \ \textbf{Problem} $\mathbf{P_1}$: Determine whether $F = 0$ is implied by 
\[
\begin{array}{ll}
& C_i=0,\ i=1,\ldots,r_1,\\
& C_j\ge0,\ j=r_1+1,\ldots,r_2,\\
& C_k\ge0,\ k=r_2+1,\ldots,r_3.
\end{array}
\]

Step 3. Apply Algorithm 1 to reduce  $\{C_l,l\in\mathcal{N}_{r_3}\backslash\mathcal{N}_{r_1}\}$ by 
$\{C_l=0,l\in\mathcal{N}_{r_1}\}$ to obtain the Jordan normal form of \\
\indent\indent $\{C_l,l\in\mathcal{N}_{r_1}\}$,
denoted by $B$, and 
the remainder set, denoted by $\mathbf{C}_1=\{g_i,i\in\mathcal{N}_r\}$.

Step 4. Apply Algorithm \ref{reducemcs-of-S_f} to obtain the reduced minimal characterization set of $\mathcal{R}(\mathbf{C}_1)$, denoted by \\
\indent\indent  $S_M=\widetilde{E}\cup S_{r'}$. 
%Write $S_{r'}=\{\mathbb{C}_j\ge0,j\in\mathcal{N}_{t_2}\}$.

Step 5. Let $G=\widetilde{\mathcal{R}}^{-1}(\widetilde{E})\cup B$ and compute the Jordan normal form of $G$, denoted by $\mathcal{B}=\{\mathcal{C}_{i},i\in\mathcal{N}_{t_1}\}$.

%Step 6. Reduce $F$ by $\mathcal{B}$ to obtain the remainder $F_1$. 

\noindent
// \ The original problem ${\rm P}_1$ has been transformed into \\
\noindent
// \ \textbf{Problem $\mathbf{P_2}$}: Determine whether ${F} = 0$ is implied by  $\widetilde{R}(\mathcal{B})$.
%\[
%\begin{array}{ll}
%& \mathcal{C}_i=0,\ j=1,\ldots,t_1.
%\end{array}
%\]

Step 6. Reduce $F$ by $\widetilde{R}(\mathcal{B})$ to obtain remainder $F_1$.
%Use the reduced standard basis $\mathcal{B}$ to reduce the objective inequality polynomial $F$ and get the reduced polynomial $F_1$.
If $F_1 \equiv 0$, then the objective identity $\bar{F} = 0$ is proved. \\
\indent\indent Otherwise, declare `Not Provable'.

\noindent 
// \ As explained in Procedure~I, $F_1$ involves only the free variables in the Jordan normal form $\mathcal{B}$. Therefore, \\
// \ if $F_1 \not\equiv 0$, the free variables 
can be chosen such that $F_1$ is evaluated to a nonzero value.

\begin{remark}
An information identity $F = 0$ is equivalent to the two information inequalities
$F \ge 0$ and $F \le 0$. In the previous approach, in order to prove $F = 0$,
$F \ge 0$ and $F \le 0$ are proved separately by solving two LPs.
In Procedure II, we transform the proof into a Gauss elimination problem, which greatly reduces the computational complexity.
\end{remark}

\begin{remark}
Procedures I and II can be implemented on the computer by Maple
for symbolic computation. 
Therefore, they can give explicit proofs of information inequalities and identities.
%At the same time, symbolic computation is used to reduce the problem, which greatly reduce the complexity of the computation.
\end{remark}

%Step 12. In this step, we also can only select the constraints $C_i$s' which corresponds to $F_1$.

%Step 12. Let $F_2=F_1-\sum_{i=1}^{t_2}p_i\mathbb{C}_i$. If we want the polynomial $F_2=0$, then we need all the coefficients of the variables (in $Var$) in $F_2$ are zero. Thus, we get a group of linear equations $P_1,..,P_m$, which are described in $p_i, i=1,..,t_2$.

%Step 13. Solve this group of linear equations to get the solution set of $p_i, i=1,..,t_2$, which contains all the possibilities of $p_i$ satisfying the condition.

%Step 14. Make all the $p_i\geq0$, and solve the inequality equations' group, then get all the values of $p_i$.

%Step 15. If we can find this kind of $p_i$, the inequality is proved. If we can not find $p_i$, then the inequality can not be proved in Shannon element inequalities.

%%%%-----------------------------------------------------------------------------------------------------------------------------------------

\section{Illustrative examples}

In this section, we give two examples to illustrate Procedures I and II. The computation is performed by Maple.

\subsection{Information Inequality under Equality Constraints}

\begin{example}\label{example-1}
$I(X_i;X_4)=0,\ i=1,2,3$ and $H(X_4|X_i,X_j)=0, 1\leq i<j\leq 3$ $\Rightarrow$ $H(X_i)\geq H(X_4)$.
\end{example}

\begin{proof}
By symmetry of the problem, we only need to prove $H(X_1)\geq H(X_4)$. The proof is given according to Procedure~I.

\noindent
\textbf{Input:} \\
Objective information inequality: $\bar{F}=H(X_1)- H(X_4)\geq0$.\\
Equality Constraints: $\bar{C}_{1}=I(X_1;X_4)=0$, $\bar{C}_{2}=I(X_2;X_4)=0$, $\bar{C}_{3}=I(X_3;X_4)=0$,\\
\indent\indent $\bar{C}_{4}=H(X_4|X_1,X_2)=0$, $\bar{C}_5=H(X_4|X_1,X_3)=0$, $\bar{C}_6=H(X_4|X_2,X_3)=0$.\\
$28$ element information inequalities: $\bar{C}_{k}\ge0,\ k\in\mathcal{N}_{34}\backslash\mathcal{N}_{6}$.

Step 1. The $s$-variables set contains $15$ elements.
%$S_4=\{s_{1, 1, 1, 1}, s_{1, 1, 1, 4}, s_{1, 1, 3, 1}, s_{1, 1, 3, 4}, s_{1, 2, 1, 1}, s_{1, 2, 1, 4}, s_{1, 2, 3, 1}, s_{1, 2, 3, 4}, s_{2, 2, 2, 2}, s_{2, 2, 2, 4}, s_{2, 2, 3, 2}, s_{2, 2, 3, 4}, s_{3, 3, 3, 3}, s_{3, 3, 3, 4}, s_{4, 4, 4, 4}\}$, which contains $15$ variables.
The $s$-variable sequence $\mathcal{S}_4=[s_{1, 2, 3, 4}, s_{1, 1, 3, 4}, s_{1, 2, 1, 4},\\
  s_{1, 2, 3, 1}, s_{2, 2, 3, 4}, s_{1, 1, 1, 4}, s_{1, 1, 3, 1}, s_{1, 2, 1, 1}, s_{2, 2, 2, 4}, s_{2, 2, 3, 2},
 s_{3, 3, 3, 4}, s_{1, 1, 1, 1}, s_{2, 2, 2, 2}, s_{3, 3, 3, 3}, s_{4, 4, 4, 4}]$.

Step 2. We have
$F=s_{1, 1, 1, 1}+s_{1, 1, 3, 1}+s_{1, 2, 1, 1}+s_{1, 2, 3, 1}-s_{2, 2, 2, 4}-s_{2, 2, 3, 4}-s_{3, 3, 3, 4}-s_{4, 4, 4, 4}$,\\
$C_1=s_{1, 1, 1, 4}+s_{1, 1, 3, 4}+s_{1, 2, 1, 4}+s_{1, 2, 3, 4}$, 
$C_2=s_{1, 2, 1, 4}+s_{1, 2, 3, 4}+s_{2, 2, 2, 4}+s_{2, 2, 3, 4}$, 
$C_3=s_{1, 1, 3, 4}+s_{1, 2, 3, 4}+s_{2, 2, 3, 4}+s_{3, 3, 3, 4}$,\\
$C_4=s_{3, 3, 3, 4}+s_{4, 4, 4, 4}$,
$C_5=s_{2, 2, 2, 4}+s_{4, 4, 4, 4}$, 
$C_6=s_{1, 1, 1, 4}+s_{4, 4, 4, 4}$,
and $28$ linear polynomials $C_{k},k\in\mathcal{N}_{34}\backslash\mathcal{N}_6$ are obtained from the $28$ element information inequalities.

%Step 5. Denote the equality polynomial as $C_{i},i=1,\ldots,6$ and the nonnegative polynomial as $C_{j},j=7,\ldots,34$.
%Step 6. Reorder the variables: $\mathcal{S}_4:=[s_{1, 2, 3, 4}, s_{1, 1, 3, 4}, s_{1, 2, 1, 4}, s_{1, 2, 3, 1}, s_{2, 2, 3, 4}, s_{1, 1, 1, 4}, s_{1, 1, 3, 1}, s_{1, 2, 1, 1}, s_{2, 2, 2, 4}, s_{2, 2, 3, 2}, s_{3, 3, 3, 4}, s_{1, 1, 1, 1}, s_{2, 2, 2, 2}, s_{3, 3, 3, 3}, s_{4, 4, 4, 4}]$.

Step 3. Compute the Gauss-Jordan normal form of $\{ C_i, i \in {\cal N}_6 \}$, $B=\{s_{3, 3, 3, 4}+s_{4, 4, 4, 4}, \, s_{2, 2, 2, 4}+s_{4, 4, 4, 4}, \, s_{1, 1, 1, 4}+s_{4, 4, 4, 4}, \, s_{1, 2, 1, 4}-s_{2, 2, 3, 4}, \, s_{1, 1, 3, 4}-s_{2, 2, 3, 4}, \, s_{1, 2, 3, 4}+2s_{2, 2, 3, 4}-s_{4, 4, 4, 4}\}$. Use 
Algorithm~\ref{division-algorithm} to reduce $\{C_l,l\in\mathcal{N}_{34}\backslash\mathcal{N}_6\}$ by $\widetilde{R}(B)$ to obtain the remainder set
$\textbf{C}_1=\{g_i,i\in\mathcal{N}_{18}\}$.

%Step 5. From the \textbf{Implied Equality Algorithm}, we have $E_0=\sum_{i=1}^{18}{v_ig_i}$, and
%\begin{equation*}\begin{array}{ll}
%V_{1} = V_{4}=V_{6}=0, V_{2} = v_{10}+v_{11}-v_{13}+v_{17}+v_{18}, V_{3} = -v_{10}-v_{14}-v_{17},\\
% V_{5} = -v_{11}+v_{13}+v_{14}+v_{15}-v_{16}-v_{17}-2v_{18},
%V_{7} = v_{8}+v_{15}-v_{16}-v_{17}-v_{18}, V_{8} = v_{8},\\
% V_{9} = -v_{10}-v_{11}-v_{16}-v_{17}-v_{18}, V_{10} = v_{10}, V_{11} = v_{11},
%V_{12} = -v_{13}-v_{14}-v_{15}+v_{16}+v_{17}+v_{18},\\
% V_{13} = v_{13}, V_{14} = v_{14}, V_{15} = v_{15}, V_{16} = v_{16}, V_{17} = v_{17}, V_{18} = v_{18}. 
%\end{array}\end{equation*}
%
%Let $\mathcal{V}_i:=subs(S_{v},v_{i})$.
%
%Because $V_1=V_4=V_6=0$, we only need to solve $15$ linear programming problems:

%$L_{k}$:\ \ ${\rm max}(V_k)$ s.t. $V_{i}\geq0$ for $i\in\mathcal{N}_{18}\backslash\{1,4,6\}$. 

%$\{\mathcal{V}_i\geq0, i=1..18\}$, and obtain $v_{10} = 0, v_{11} = 0, v_{13} = 0, v_{14} = 0, v_{15} = 0, v_{16} = 0, v_{17} = 0, v_{18} = 0, v_{8} \geq0$. Without generality, we set $v_{8}=1$, so we have $v_7=v_8=1$ with other $v_{i}$'s are zero.

%We obtain the optimal values of $L_{k}$ are equal to $0$ for $k\in\mathcal{N}_{18}\backslash\{7,8\}$, and the optimal values of $L_{k}$ are equal to $+\infty$ for $k=7,8$. Thus we have $E_C=\{g_{7}=0,g_8=0\}$.

Step 4. Use Algorithm \ref{reducemcs-of-S_f} to obtain $S_M=\widetilde{E}\cup S_{r'}$ 
and $S_{r'}=\{\mathbb{C}_i, i=1,\ldots,10\}$, where 
$$\begin{array}{l}
\mathbb{C}_1=s_{1, 1, 1, 1}, \, \mathbb{C}_2=s_{1, 1, 3, 1}, \, \mathbb{C}_3=s_{1, 2, 1, 1}, \, \mathbb{C}_4=s_{2, 2, 2, 2}, \, \mathbb{C}_5=s_{2, 2, 3, 2}, \, \mathbb{C}_6=s_{2, 2, 3, 4}, \, \mathbb{C}_7=s_{3, 3, 3, 3},\\
\mathbb{C}_{8}=s_{1, 2, 3, 1}-s_{2, 2, 3, 4}+s_{1, 1, 3, 1}, \, \mathbb{C}_{9}=s_{1, 2, 3, 1}-s_{2, 2, 3, 4}+s_{1, 2, 1, 1},\, \mathbb{C}_{10}=s_{1, 2, 3, 1}-s_{2, 2, 3, 4}+s_{2, 2, 3, 2}.
\end{array}$$
% which are all sorted by the order given in $\mathcal{S}_4$.$

Step 5. Compute the Gauss-Jordan normal form $\mathcal{B}=\{s_{4, 4, 4, 4}, \, s_{3, 3, 3, 4}, \, s_{2, 2, 2, 4}, \, s_{1, 1, 1, 4}, \, s_{1, 2, 1, 4}-s_{2, 2, 3, 4},$ \\
$s_{1, 1, 3, 4}-s_{2, 2, 3, 4}, \, s_{1, 2, 3, 4}+2s_{2, 2, 3, 4}\}$.

Step 6. Reduce $F$ by $\widetilde{R}(\mathcal{B})$ to obatain $F_1=s_{1, 1, 1, 1}+s_{1, 2, 1, 1}-s_{2, 2, 3, 4}+s_{1, 1, 3, 1}+s_{1, 2, 3, 1}$.

%Step 8. We have $\mathbf{C}_2=\{s_{1, 1, 1, 1},s_{1, 1, 3, 1},s_{1, 2, 1, 1},s_{2, 2, 2, 2},s_{2, 2, 3, 2},s_{2, 2, 3, 4},s_{3, 3, 3, 3},s_{1, 1, 3, 1}+s_{1, 2, 3, 1},s_{1, 2, 1, 1}+s_{1, 2, 3, 1},s_{1, 2, 3, 1}+s_{2, 2, 3, 2},s_{2, 2, 3, 2}+s_{2, 2, 3, 4},s_{2, 2, 3, 4}+s_{1, 1, 3, 1},s_{2, 2, 3, 4}+s_{1, 2, 1, 1}, s_{1, 2, 3, 1}-s_{2, 2, 3, 4}+s_{1, 1, 3, 1},s_{1, 2, 3, 1}-s_{2, 2, 3, 4}+s_{1, 2, 1, 1},s_{1, 2, 3, 1}-s_{2, 2, 3, 4}+s_{2, 2, 3, 2}\}$.

%Step 9. We obtain $\mathbf{C}_3=\{\mathbb{C}_1=s_{1, 1, 1, 1},\mathbb{C}_2=s_{1, 1, 3, 1},\mathbb{C}_3=s_{1, 2, 1, 1},\mathbb{C}_4=s_{2, 2, 2, 2},\mathbb{C}_5=s_{2, 2, 3, 2},\mathbb{C}_6=s_{2, 2, 3, 4},\mathbb{C}_7=s_{3, 3, 3, 3},
%
%\mathbb{C}_8=s_{1, 1, 3, 1}+s_{1, 2, 3, 1},\mathbb{C}_9=s_{1, 2, 1, 1}+s_{1, 2, 3, 1},\mathbb{C}_{10}=s_{1, 2, 3, 1}+s_{2, 2, 3, 2},
%\mathbb{C}_{11}=s_{2, 2, 3, 2}+s_{2, 2, 3, 4},\mathbb{C}_{12}=s_{2, 2, 3, 4}+s_{1, 1, 3, 1},\mathbb{C}_{13}=s_{2, 2, 3, 4}+s_{1, 2, 1, 1},
%
%\mathbb{C}_{8}=s_{1, 2, 3, 1}-s_{2, 2, 3, 4}+s_{1, 1, 3, 1},\mathbb{C}_{9}=s_{1, 2, 3, 1}-s_{2, 2, 3, 4}+s_{1, 2, 1, 1},\mathbb{C}_{10}=s_{1, 2, 3, 1}-s_{2, 2, 3, 4}+s_{2, 2, 3, 2}\}$.
% which are all sorted by the order given in $\mathcal{S}_4$.

Steps 7-11. We have $t_2=10$, $n_1=8$, $\bar{S}_{P}=\{p_{9}+p_{10}, 1-p_{9},  -p_{10}, 1-p_{9}-p_{10}, p_{9}, p_{10}\}$ and $S_P=\bar{S}_{P}\cup\{1,0\}$.
%Let $F_2=F_1-\sum_{i=1}^{10}p_i\mathbb{C}_{i}$, where $P=\{p_i,i\in\mathcal{N}_{10}\}$, and $Q=\{q_j=0,j\in\mathcal{N}_8\}$.
%and get the linear solution equations' group from $F_2=0$: $\{p_{1} = 1, p_{2} = p_{9}+p_{10}, p_{3} = 1-p_{9}, p_{4} = 0, p_{5} = -p_{10}, p_{6} = 0, p_{7} = 0, p_{8} = 1-p_{9}-p_{10}, p_{9} = p_{9}, p_{10} = p_{10}\}$.

Step 12. Solve the LP in Problem ${\rm P}_3$ to complete the proof. Alternatively, we can solve the inequality set $\mathcal{R}(\bar{S}_P)$ to obtain the solution $\{0 \leq p_{9} \leq 1,\ p_{10} = 0\}$. Substituting $p_9=0$ and $p_{10}=0$ to $\{p_i=P_i,i\in\mathcal{N}_{10}\}$ yields $\{p_{1} = 1, p_{2} = 0, p_{3} = 1, p_{4} = 0, p_{5} = 0, p_{6} = 0, p_{7} = 0, p_{8} = 1, p_{9} = 0, p_{10} = 0\}$. 
%Make all the $p_i\geq0$, and solve the inequality equations' group, then get all the values of $p_i$: $\{0 \leq p_{9} \leq 1,\ p_{10} = 0\}$.
Thus an explicit proof is given by $F_1=\mathbb{C}_1+\mathbb{C}_3+\mathbb{C}_{8}\geq0$.
%Step 17. The inequality can be proved with $\{p_{1} = 1, p_{2} = 0, p_{3} = 1, p_{4} = 0, p_{5} = 0, p_{6} = 0, p_{7} = 0, p_{8} = 1, p_{9} = 0, p_{10} = 0\}$. So $F_1=\mathbb{C}_1+\mathbb{C}_3+\mathbb{C}_{8}\geq0$.
\end{proof}

\begin{remark}

Table \ref{table-compare} shows the advantage of Procedure~I for Example \ref{example-1} by comparing it with the Direct LP method induced by Theorem~\ref{LP-S}.
\begin{table}[h!]
\caption{}
\label{table-compare}
\begin{center}
\begin{tabular}{ |c|c|c|c| } 
 \hline
  & Number of variables & Number of equality constraints & Number of Inequality constraints \\ 
 \hline
Direct LP method & 15 & 6 & 28 \\ 
 \hline
% $\mathbf{L}_{R}$ & 8 & 0 & 10 \\ 
% \hline
LP in Problem ${\rm P}_3$ & 2 & 0 & 6 \\ 
 \hline
\end{tabular}
\end{center}
\end{table}

\end{remark}

\subsection{Information Identity under Equality Constraints}

\begin{example}
$I(X_1;X_2|X_3)=0$, $H(X_3)=I(X_2;X_3|X_1)$ $\Rightarrow$ $H(X_1)=H(X_1|X_2,X_3)$.
\end{example}

\begin{proof}
The proof is given according to Procedure II.

\noindent
\textbf{Input:} \\
Objective information inequality: $\bar{F}=H(X_1)-H(X_1|X_2,X_3)\geq0$.\\
Equality Constraints: $\bar{C}_{1}=I(X_1;X_2|X_3)=0$, $\bar{C}_{2}=H(X_3)-I(X_2;X_3|X_1)=0$.\\
$9$ element information inequalities: $\bar{C}_{k}\ge0,\ k\in\mathcal{N}_{11}\backslash\mathcal{N}_{2}$.

Step 1. The $s$-variables set contains $7$ elements. The $s$-variable sequence
$\mathcal{S}_3= $ 
$[s_{1,2,3}, s_{1,1,3}, s_{1,2,1}, s_{2,2,3}, s_{1,1,1}, s_{2,2,2}, s_{3,3,3}]$.
%$S_3=\{s_{1,1,1}, s_{1,1,3}, s_{1,2,1}, s_{1,2,3}, s_{2,2,2}, s_{2,2,3}, s_{3,3,3}\}$.

Step 2. We have $F=s_{1,1,3}+s_{1,2,1}+s_{1,2,3}$,
$C_1=s_{1,2,1}$, $C_2=s_{1,1,3}+s_{1,2,3}+s_{3,3,3}$,
%and the other $9$ linear polynomials corresponding to the $9$ element information inequalities:
$C_3=s_{1,1,1}$, $C_4=s_{2,2,2}$, $C_5=s_{3,3,3}$, $C_6=s_{1,2,1}+s_{1,2,3}$, $C_7=s_{1,2,3}+s_{2,2,3}$, $C_8=s_{1,1,3}+s_{1,2,3}$, $C_9=s_{1,2,1}$, $C_{10}=s_{1,1,3}$ and $C_{11}=s_{2,2,3}$.

%Rewrite $H(X_1)-H(X_1|X_2,X_3)\overset{\Delta}{=}s_{1,1,1}+s_{1,1,3}+s_{1,2,1}+s_{1,2,3}-s_{1,1,1}=s_{1,1,3}+s_{1,2,1}+s_{1,2,3}=F$.

%Step 3. Rewrite the additional constraints as $I(X_1;X_2|X_3)\overset{\Delta}{=}s_{1,2,1}=0$, $H(X_3)-I(X_2;X_3|X_1)\overset{\Delta}{=}s_{1,1,3}+s_{1,2,3}+s_{2,2,3}+s_{3,3,3}-s_{2,2,3}=s_{1,1,3}+s_{1,2,1}+s_{1,2,3}=0$.

%Step 4. The element inequality polynomials (nonnegative) are rewritten as: $s_{1,1,1}$, $s_{2,2,2}$, $s_{3,3,3}$, $s_{1,2,1}+s_{1,2,3}$, $s_{1,2,3}+s_{2,2,3}$, $s_{1,1,3}+s_{1,2,3}$, $s_{1,2,1}$, $s_{1,1,3}$, $s_{2,2,3}$.

%Step 5. Denote the linear equality polynomials as $C_{1}=s_{1,2,1}$, $C_{2}=s_{1,1,3}+s_{1,2,1}+s_{1,2,3}$, the linear inequality polynomials as $C_3=s_{1,1,1}$, $C_4=s_{2,2,2}$, $C_{5}=s_{3,3,3}$, $C_6=s_{1,2,1}+s_{1,2,3}$, $C_7=s_{1,2,3}+s_{2,2,3}$, $C_{8}=s_{1,1,3}+s_{1,2,3}$, $C_{9}=s_{1,2,1}$, $C_{10}=s_{1,1,3}$, $C_{11}=s_{2,2,3}$.

%Step 6. Reorder the new variables: $\mathcal{S}_3=[s_{1,2,3}, s_{1,1,3}, s_{1,2,1}, s_{2,2,3}, s_{1,1,1}, s_{2,2,2}, s_{3,3,3}]$.

Step 3. Compute the Gauss-Jordan normal form $B=\{s_{1,2,1},s_{1,1,3}+s_{1,2,3}+s_{3,3,3}\}$.
Use Algorithm~\ref{division-algorithm} to reduce $\{C_l,l\in\mathcal{N}_{11}\backslash\mathcal{N}_2\}$ by $\widetilde{R}(B)$ to obtain the remainder set $\mathbf{C}_1=\{g_i,i\in\mathcal{N}_{8}\}$, where $g_{1} = s_{1, 1, 1}, \, g_{2} = s_{2, 2, 2}, \, g_{3} = s_{3, 3, 3}, \, g_{4} = -s_{1, 1, 3}-s_{3, 3, 3}, \, g_{5} = -s_{1, 1, 3}-s_{3, 3, 3}+s_{2, 2, 3}, \, g_{6} = -s_{3, 3, 3}, \, g_{7} = s_{1, 1, 3}, \, g_{8} = s_{2, 2, 3}$.

Step 4. Use Algorithm \ref{reducemcs-of-S_f} to obtain $S_M=\widetilde{E}\cup S_{r'}$, where $\widetilde{E}=\{s_{1,1,3}=0,\ s_{3,3,3}=0\}$.

Step 5. Compute the Gauss-Jordon normal form $\mathcal{B}=\{s_{1, 2, 3}, s_{1, 1, 3}, s_{1, 2, 1}, s_{3, 3, 3}\}$.

Step 6. Reduce $F$ by $\mathcal{B}$ to obtain $F_1 \equiv 0$. Thus the information identity is proved.

\end{proof}

\section{Conclusion and discussion}
\label{sec-conc}
In this paper, we develop a new method to prove linear information inequalities and identities. Instead of solving an LP, we transform the problem into a polynomial reduction problem. For the proof of information inequalities, compared with existing  methods (ITIP and its variations), our method takes advantage of the algebraic structure of the problem and greatly reduces the computational complexity. For the proof of information identities, we give a simple direct proof method which is much more efficient than existing methods.

\appendices

\section*{Acknowledgment}
This work is partially supported by NSFC 11688101, and Fundamental Research Funds for the Central Universities (2021NTST32).

\end{document}